\documentclass[11pt]{article}
\usepackage[letterpaper]{geometry}
\usepackage{amsmath,amsthm,amsfonts,amssymb}
\usepackage{enumerate,color,xcolor}
\usepackage{graphicx}
\usepackage{subfigure}
\usepackage{url}

\numberwithin{equation}{section}
\theoremstyle{plain}

\newtheorem{theorem}{Theorem}

\newtheorem{definition}[theorem]{Definition}

\newtheorem{proposition}[theorem]{Proposition}

\newtheorem{remark}[theorem]{Remark}
\newtheorem{assumption}[theorem]{Assumption}

\usepackage{comment}

\usepackage{amsmath,amssymb,amsthm,dsfont}
\numberwithin{equation}{section}
\usepackage[colorlinks,linkcolor=orange,citecolor=blue]{hyperref}

\begin{document}

\begin{center}
  \Large \bf Short-maturity options on realized variance in local-stochastic volatility models
\end{center}

\author{}
\begin{center}
{Dan Pirjol}\,\footnote{School of Business, Stevens Institute of Technology, United States of America;
  dpirjol@gmail.com},
  {Xiaoyu Wang}\,\footnote{Hong Kong University of Science and Technology (Guangzhou), People's Republic of China;
  xiaoyuwang@hkust-gz.edu.cn},
  Lingjiong Zhu\,\footnote{Department of Mathematics, Florida State University, United States of America; zhu@math.fsu.edu
 }
\end{center}

\begin{center}
 \today
\end{center}

\begin{abstract}
We derive the short-maturity asymptotics for prices of options on realized variance in local-stochastic volatility models. We consider separately the short-maturity asymptotics for out-of-the-money and in-the-money options cases. 
The analysis for the out-of-the-money case uses large deviations theory and the solution for the rate function
involves solving a two-dimensional variational problem. 
In the special case when the Brownian noises
in the asset price dynamics and the volatility process
are uncorrelated, we solve this variational problem explicitly.
For the correlated case, we obtain upper and lower bounds 
for the rate function, as well as an expansion around the at-the-money point.
Numerical simulations of the prices of variance options in a local-stochastic volatility model with bounded local volatility are in good agreement with the asymptotic results for sufficiently small maturity.
The leading-order asymptotics for at-the-money options on realized variance is dominated by fluctuations of the asset price around the spot value, and is computed in closed form.
\end{abstract}

\section{Introduction}

Options on realized variance are derivative contracts whose payoff is linked to the annualized realized variance of the return of some asset, which can be a stock, index, interest rate, exchange rate, or futures on some asset. They are related to variance swaps, which are instruments that pay an amount equal to the realized variance at maturity.

Denoting the price of the asset on a set of discrete sampling time points $\{S_k\}_{k=0}^n$ which are uniformly spaced $0=t_0< t_1 < \cdots < t_n=T$, i.e. $t_{i}-t_{i-1}=\tau$ for every $i=1,2,\ldots,n$, where $\tau$ is the time step of the sampling period expressed in years (for daily sampling $\tau=\frac{1}{252}$). 
The annualized realized variance is given by
\begin{equation}
\mathrm{RV}_{T,n} = \frac{1}{n\tau} \sum_{i=1}^n \log^2 ( S_i/S_{i-1})\,.
\end{equation}
Denote the discrete-time sum 
\begin{equation}
P_n(T) = \sum_{i=1}^n \log^2 ( S_i/S_{i-1})\,.
\end{equation}
In practice, this sum is approximated, for sufficiently large $n$, with the quadratic variation of the 
log-asset price
\begin{equation}
P(T) = [\log S, \log S]_T\,.
\end{equation}
If the asset price is assumed to follow a diffusion of the form $dS_t/S_t = \sigma_t dW_t + (r-q) dt$, where $\sigma_t$ is an arbitrary stochastic process and
$W_{t}$ is a standard Brownian motion, the quadratic variation of the log-price is 
\begin{equation}
[\log S, \log S]_T = \int_0^T \sigma^2_s ds \,.
\end{equation}
It was noted by Jarrow et al. (2013) \cite{Jarrow2013} that the limit of the expectation $\mathbb{E}[P_n(T)]$ of the discrete-time approximation does not always coincide with the expectation of the continuous time quantity $\mathbb{E}[P(T)]$, since convergence in probability does not imply convergence in $L_1$ norm. For example, this does not hold for the 3/2 model of stochastic volatility model for certain values of the model parameters. 

We will consider in this paper the continuous-time limit, and will use for the underlying of the variance swaps and options with the quadratic variation of $\log S_t$. One must keep in mind that the convergence of the discrete-time result to the continuous time counterpart will have to be checked for each case.

In this paper, we study the short-maturity asymptotics for variance options in Markovian local-stochastic volatility models. 
The short maturity asymptotics of options on realized variance and options on VIX has been studied by Al\'os, Garcia-Lor\'{i}te and Gonzalez \cite{Alos2018} and by Lacombe, Muguruza and Stone \cite{Lacombe2021} in a class of multi-factor rough volatility models. 

Another case which is implicitly covered in the literature is that of variance options under Markovian stochastic volatility models, which are equivalent to Asian options under local volatility models. The short maturity asymptotics of these options have been studied, for both out-of the-money and at-the-money regimes, in \cite{PZAsian}.
We will further elaborate on the connection between Asian options in local volatility models and variance options in Section~\ref{sec:local:stoch:vol}.

To the best of our knowledge, the short-maturity asymptotics for variance options in local-stochastic volatility models has never been rigorously studied in the literature. In this paper, we present explicit results for the leading-order asymptotics, considering both cases of 
out-of-the-money and at-the-money variance options.

Recently, CBOE has started trading futures on realized variance since Monday 23 September 2024. These futures are called \texttt{Cboe S\&P 500 Variance Futures} (Ticker: VA) and are cash-settled futures contracts based on the realized variance of the S\&P 500 
index.\footnote{https://www.cboe.com/variance-futures-pipeline-hub/} 
They will reflect the market view of prices of cash-settled variance swaps. Options on these futures contracts are similar to options on realized variance. Our results can be applied directly to the pricing of these options.

\subsection{Literature review}

Pricing options on realized variance has been discussed in the literature under several types of models. We give next a brief literature review.

A class of popular models with practitioners for pricing variance products are forward variance
models. Such models describe the dynamics of the instantaneous forward variance. 
An example of such models are the Bergomi models \cite{BergomiBook,Bergomi1, Bergomi2}.
The pricing of variance options under the Bergomi models was discussed in \cite{Bergomi1,Bergomi2}.

Pricing of variance options in the Heston model requires the distribution of the time integral of a CIR process $\int_0^T V_s ds$ where
$dV_t = \kappa (\theta - V_t) dt + \xi \sqrt{V_t} dZ_t$. The Laplace transform of the density of
$\int_0^T V_s ds$ is known in closed form; see e.g. \cite{Dufresne2001}. 
Therefore standard transform methods can be used to price the variance options in this model; see for example Sepp (2008) \cite{Sepp2008}
and Sepp (2008) \cite{Sepp2008HestonJumps} which also allow jumps.

Carr, Geman, Madan and Yor (2005) \cite{CGMY2005} proposed a method for pricing options on realized variance in exponential L\'evy models using a Laplace transform method. 
Carr and Itkin (2009) \cite{Itkin2009} proposed a new asymptotic method for pricing variance and volatility swaps and options on these swaps which yields a closed-form
expression for the fair price of these instruments, if the underlying process is modeled by a L\'evy
process with stochastic time change.
Carr and Lee (2010) \cite{CarrLee2010} obtained robust model-free hedges and price bounds for options on the realized variance of the returns on an underlying price process
that is a positive continuous semimartingale.
Kallsen, Muhle-Karbe and Vo{\ss} (2011) \cite{Kallsen2011} studied the pricing of variance options and determined semi-explicit formulas in general affine models allowing for jumps, stochastic volatility, and the leverage effect using the Laplace transform approach.
Drimus (2012) \cite{Drimus2012} discussed the pricing and hedging of options on realized variance in the 3/2 stochastic volatility model using transform-based methods. 
Drimus (2012) \cite{Drimus2012LogOU} studied the pricing of options on realized variance in a general class of Log-OU stochastic volatility models.
Torricelli (2013) \cite{Torricelli2013} studied joint pricing on an asset and its realized variance for various stochastic volatility models.


Although the valuation of variance swaps and options on realized variance is more convenient in continuous time, in actual applications these instruments are defined in discrete time with daily time sampling. Thus, the discrete-time case received a lot of attention in the literature.
Sepp (2012) \cite{Sepp2012} 
analyzed the effect of discrete sampling on the valuation of options on the realized variance in the Heston model. 

The systematic bias of the time discretization and the asymptotics for variance swaps was studied for a few popular stochastic volatility models by Bernard and Cui (2013) \cite{Bernard2013}.
Keller-Ressel and Muhle-Karbe (2013) \cite{KR2013} found that the difference between options on discretely sampled realized variance and the continuous time limit 
strongly depends on whether or not the stock price process has jumps. They proposed an approximation method based on correcting prices of options on quadratic variation by their asymptotic results, and an exact method using a novel randomization approach and applying Fourier-Laplace techniques.
Lian, Chiarella and Kalev (2014) \cite{Lian2014} obtained an accurate approximation for the characteristic function of the discretely sampled realized variance, 
which yielded semi-analytical pricing formulae for variance options and other derivatives.
Zheng and Kwok (2014) \cite{Zheng2014} used a saddlepoint approximation method to price options on discrete realized variance, 
and the same authors derived in \cite{Zheng2014MF} closed-form pricing formulas for discretely sampled variance swaps.
Zheng and Kwok (2014) \cite{Zheng2014Fourier} developed efficient fast Fourier transform algorithms to price and hedge
options on discrete realized variance and other products under time-inhomogeneous L\'{e}vy processes.
Zheng and Kwok (2015) \cite{Zheng2015} used the partially exact and bounded approximations to derive efficient and accurate analytic approximation formulas for pricing options on discrete realized variance under affine stochastic volatility models with jumps.
Zheng, Yuen and Kwok (2016) \cite{Zheng2016} developed recursive algorithms for pricing pricing variance options and volatility swaps on discrete realized variance under general time-changed L\'{e}vy processes.
Drimus, Farkas and Gourier (2016) \cite{Drimus2016} studied the valuation of options on discretely sampled variance
by analyzing the discretization effect and obtaining an analytical correction term to be applied to the value of options on continuously sampled variance
under general stochastic volatility dynamics.
Cui, Kirkby, and Nguyen (2017) \cite{Cui2017} developed a transform-based method to price swaps and options related to discretely-sampled realized variance under a general class of stochastic volatility models with jumps.
A survey on recent results in the pricing of derivatives on discrete realized variance can be found in Kwok and Zheng (2022) \cite{Kwok2022}.

Moreover, various properties for variance options have been studied in the literature.
Carr, Geman, Madan and Yor (2011) \cite{CGMY2011} analyzed the property of monotonicity in maturity for call options at a fixed strike
for realized variance option and options on quadratic variation normalized to unit.
Griessler and Keller-Ressel \cite{Griessler2014} showed that options on variance are typically underpriced if quadratic variation is substituted for the discretely sampled realized variance for a class of models including independently time-changed L\'{e}vy models and Sato processes with symmetric jumps.

Finally, our work is related to the short-maturity asymptotics for path-dependent option prices that have been studied in the recent literature.
The main tools are the sample-path large deviation principle
for small-time diffusion processes that dates back to \cite{Varadhan} 
and the contraction principle from large deviations theory (see e.g. \cite{VaradhanLD,Dembo1998}).
Most of such works have been focused on the
short-maturity asymptotics for Asian options.
The first short-maturity asymptotics result
for Asian options was obtained in \cite{PZAsian}
for local volatility models. 
The out-of-the-money (OTM) case relies
on large deviations for small-time diffusion processes, and the rate function is a one-dimensional variational problem
that can be solved explicitly \cite{PZAsian}.
Similar studies have been carried out for Asian options
under the CEV model \cite{PZAsianCEV}.
The short-maturity forward start Asian option has been studied in \cite{forwardAsian}.
By a combination of a Gaussian process approximation and Malliavin calculus, 
\cite{Park2019} studied 
both pricing and hedging for short-maturity
Asian options for local volatility models.

\subsection{Summary of the paper}

The outline of the paper is as follows.
In Section~\ref{sec:local:stoch:vol} we introduce the model for the underlying and specify technical assumptions on the model parameters.
Section~\ref{sec:main:results} contains the main results.
The short-maturity asymptotics for OTM
variance options are given in Section~\ref{sec:OTM}.
We show that by using large deviations theory,
the leading-order term in the short-maturity asymptotics for OTM
variance options can be formulated
as a two-dimensional variational problem.
In Section~\ref{sec:variational} we discuss the solution of this variational problem.
For the particular case when the Brownian noises
in the asset price dynamics and the volatility process
are uncorrelated, we solve the variational problem for the rate function explicitly
by reducing it to a one-dimensional optimization problem.
The argument of this optimization problem includes the rate function for
short-maturity Asian options under local volatility models obtained in \cite{PZAsian}.
For the correlated case, we obtain upper and lower bounds
for the solution of the variational problem. 
In Section~\ref{sec:perfect}, we solve the variational problem for perfectly correlated
and anti-correlated asset and volatility.
In Section~\ref{sec:expansion:around:ATM}, we give an explicit result for the
rate function of the variance options in an expansion in log-moneyness, which is convenient for 
use in practical applications.
Finally, in Section~\ref{sec:ATM}, further short-maturity
asymptotics results are obtained when the variance options
are at-the-money (ATM).

Section~\ref{sec:appl} discusses the application of the asymptotic results to pricing
variance options in the local-stochastic volatility model.
Numerical tests of the asymptotic results are presented.
For the uncorrelated case $\rho=0$ we compare the expansion in log-moneyness with an exact solution of the variational problem. The expansion is also tested by comparing with general bounds, and by comparing with Monte Carlo simulations in a model with bounded local volatility.

Some background of large deviations theory is presented in Appendix~\ref{sec:LDP}.
The technical proofs of all the results in the main paper
are provided in Appendix~\ref{sec:proofs}.
Finally, some additional technical results are provided in Appendix~\ref{sec:g2:h2}.

\section{Model Setup}\label{sec:local:stoch:vol}

In this paper, we are interested in studying variance options under
a local-stochastic volatility model.
Suppose that under the risk-neutral probability measure $\mathbb{Q}$ the underlying asset $S_{t}$ 
and the variance process $V_{t}$ follow a local-stochastic volatility model
of the form:
\begin{align}
&\frac{dS_{t}}{S_{t}}=(r-q)dt+\eta(S_{t})\sqrt{V_{t}}\rho dZ_{t}+\eta(S_{t})\sqrt{V_{t}}\sqrt{1-\rho^{2}}dW_{t},\label{eqn:S}
\\
&\frac{dV_{t}}{V_{t}}=\mu(V_{t})dt+\sigma(V_{t})dZ_{t},\label{eqn:V}
\end{align}
where $r$ is the risk-free rate, $q$ is the dividend yield, and $W_{t},Z_{t}$ are two independent standard Brownian motions, 
where the functions $\eta(\cdot), \sigma(\cdot):\mathbb{R}^{+}\rightarrow\mathbb{R}^{+}$ and $\mu(\cdot):\mathbb{R}^{+}\rightarrow\mathbb{R}$ are 
assumed to be time-homogeneous for simplicity.

The payoff of variance options is linked to the realized variance of the asset price. In continuous time this is related to the quadratic variation of the log-price, 
which is expressed in the model (\ref{eqn:S}), (\ref{eqn:V}) as
\begin{equation}
[\log S]_T = \int_0^T V_{s}\eta^{2}(S_{s})ds \,.
\end{equation}

The fair strike of a variance swap with maturity $T$ is defined as
\begin{equation}\label{FVT}
F_V(T) = \mathbb{E}\left[ \int_0^T V_{s}\eta^{2}(S_{s})ds \right]\,.
\end{equation}

The call and put variance options prices are given by 
\begin{align}
&C(T)=e^{-rT}\mathbb{E}\left[\left(\frac{1}{T}\int_{0}^{T}V_{s}\eta^{2}(S_{s})ds-K\right)^{+}\right],
\nonumber
\\
&P(T)=e^{-rT}\mathbb{E}\left[\left(K-\frac{1}{T}\int_{0}^{T}V_{s}\eta^{2}(S_{s})ds\right)^{+}\right],
\label{option:price}
\end{align}
where $K>0$ is the strike price and $T>0$ is the maturity. 
A variance call option is out-of-the-money (OTM) if $K > F_V(T)$, in-the-money (ITM)
if $K<F_V(T)$ and at-the-money (ATM) if $K = F_V(T)$.
A variance put option is OTM if $K < F_V(T)$, ITM if $K > F_V(T)$ and ATM if $K=F_V(T)$.

We first assume that 
$\eta(\cdot),\mu(\cdot)$ and $\sigma(\cdot)$ are
uniformly bounded for simplicity. 

\begin{assumption}\label{assump:bounded}
We assume that $\eta(\cdot),\mu(\cdot)$ and $\sigma(\cdot)$ are
uniformly bounded:
\begin{equation}
\sup_{x\in\mathbb{R}^{+}}\eta(x)\leq M_{\eta},
\qquad
\sup_{x\in\mathbb{R}^{+}}|\mu(x)|\leq M_{\mu},
\qquad
\sup_{x\in\mathbb{R}^{+}}\sigma(x)\leq M_{\sigma}.
\end{equation}
\end{assumption}

In the remainder of the paper, we often further assume that $\eta(\cdot)$ is decreasing, which will be stated explicitly where this is assumed,
such that our model satisfies the leverage effect in finance.
More precisely, when $\eta(\cdot)$ is not a constant function,
we assume that $\eta(\cdot)$ is strictly decreasing
so that its inverse function $\eta^{-1}(\cdot)$ exists. 
We also provide the following assumptions about the Lipschitz continuity.

\begin{assumption}\label{assump:lip}
We assume that $\eta$ is $\ell_{\eta}$-Lipschitz and $\sigma$ is $\ell_{\sigma}$-Lipschitz.
\end{assumption}

In addition, we impose the following
assumption on $\eta(\cdot)$ and $\sigma(\cdot)$
that appear in the diffusion terms of \eqref{eqn:S}-\eqref{eqn:V}
that is needed for the small-time large deviations estimates for \eqref{eqn:S}-\eqref{eqn:V}.

\begin{assumption}\label{assump:LDP}
We assume that $\inf_{x\in\mathbb{R}^{+}}\sigma(x)>0$
and $\inf_{x\in\mathbb{R}^{+}}\eta(x)>0$. Moreover, 
there exist some constants $M,\alpha>0$ such that
for any $x,y\in\mathbb{R}^{+}$,
$|\sigma(e^{x})-\sigma(e^{y})|\leq M|x-y|^{\alpha}$
and $|\eta(e^{x})-\eta(e^{y})|\leq M|x-y|^{\alpha}$.
\end{assumption}

To satisfy the leverage effect that is commonly observed in finance,
it is often assumed in the literature that $\eta(\cdot)$ is monotonically decreasing.
That is, the larger value of the asset price $S_{t}$, the smaller value
of the volatility function $\eta(S_{t})$.

When $\eta(\cdot)\equiv 1$, the local-stochastic volatility model \eqref{eqn:S}-\eqref{eqn:V}
reduces to the stochastic volatility model:
\begin{align}\label{SVmodel}
&\frac{dS_{t}}{S_{t}}=(r-q)dt+\sqrt{V_{t}}\left(\sqrt{1-\rho^{2}}dW_{t}+\rho dZ_{t}\right),
\\
&\frac{dV_{t}}{V_{t}}=\mu(V_{t})dt+\sigma(V_{t})dZ_{t},\nonumber
\end{align}
where $r$ is the risk-free rate, $q$ is the dividend yield, and $W_{t}$ and $Z_{t}$ are two independent Brownian motions,
and for variance options, the call
and put options prices in \eqref{option:price} reduce to:
\begin{equation}
C(T)=e^{-rT}\mathbb{E}\left[\left(\frac{1}{T}\int_{0}^{T}V_{s}ds-K\right)^{+}\right],
\qquad
P(T)=e^{-rT}\mathbb{E}\left[\left(K-\frac{1}{T}\int_{0}^{T}V_{s}ds\right)^{+}\right],
\end{equation}
where $K>0$ is the strike price.
This shows that the variance option prices $C(T),P(T)$ in the stochastic volatility model \eqref{SVmodel}
are identical to those of Asian options in the local volatility model with local volatility function $\sigma(V)$.
The short-maturity asymptotics for Asian options under local volatility models have
been studied in \cite{PZAsian,PZAsianCEV,PZAsianRisk,Park2019}.

Certain models used in financial practice, such as the Heston model,  do not satisfy the Assumptions~\ref{assump:bounded} and \ref{assump:LDP}. However, our arguments can be extended to these models as well, under alternative assumptions; see Remark~\ref{rmk:Heston}.

In the rest of the paper, we focus on the variance options under the local-stochastic volatility model \eqref{eqn:S}-\eqref{eqn:V}
and study the call option price $C(T)$ and put option price $P(T)$ given in \eqref{option:price}.
Note the short-maturity limit of the fair strike of the variance swap \eqref{FVT} is
\begin{equation}\label{FVTsmallT}
F_V(0)=\lim_{T\to 0} F_V(T) = \eta^2(S_0) V_0 \,.
\end{equation}
In this paper, we are interested in the short-maturity asymptotics ($T\rightarrow 0$)
for $C(T)$ and $P(T)$. We distinguish two cases: the out-of-the-money (OTM) case
and the at-the-money (ATM) case. In the short maturity limit, by (\ref{FVTsmallT}), the moneyness is measured with respect to $F_V(0) = \eta^2(S_0) V_0$.
More explicitly, the OTM call option corresponds to $V_{0}\eta^{2}(S_{0})<K$, 
and the OTM put option corresponds to $V_{0}\eta^{2}(S_{0})>K$.
The ATM case for both call and put options corresponds to $V_{0}\eta^{2}(S_{0})=K$.
We omit the discussions for in-the-money (ITM) case here, which 
can be analyzed via put-call parity.

\section{Main Results}\label{sec:main:results}

\subsection{OTM case}\label{sec:OTM}

In the short-maturity regime, i.e. as $T\rightarrow 0$, we have $\frac{1}{T}\int_{0}^{T}V_{s}\eta^{2}(S_{s})ds$ converges
a.s. to $V_{0}\eta^{2}(S_{0})$.
Therefore, the OTM case
for variance call corresponds to $V_{0}\eta^{2}(S_{0})<K$ 
and the OTM case for variance put options corresponds to $V_{0}\eta^{2}(S_{0})>K$.
We are interested in studying the OTM short-maturity 
asymptotics for variance call and put options.
We have the following main result.

\begin{theorem}\label{thm:OTM}
Suppose that Assumptions~\ref{assump:bounded} and \ref{assump:LDP} hold, 
and that the correlation parameter is in the range $-1 < \rho < +1$.

(i) For OTM variance call options, i.e. $V_{0}\eta^{2}(S_{0})<K$, we have
\begin{equation}
\lim_{T\rightarrow 0}T\log C(T)=-\mathcal{I}_\rho(S_{0},V_{0},K),
\end{equation}
where
\begin{align}
&\mathcal{I}_\rho(S_{0},V_{0},K)
\nonumber
\\
&=\inf_{\substack{g(0)=\log S_{0},h(0)=\log V_{0}\\
\int_{0}^{1}e^{h(t)}\eta^{2}(e^{g(t)})dt=K}}
\Bigg\{\frac{1}{2(1-\rho^{2})}\int_{0}^{1}\left(\frac{g'(t)}{\eta(e^{g(t)})\sqrt{e^{h(t)}}}-\frac{\rho h'(t)}{\sigma(e^{h(t)})}\right)^{2}dt
\nonumber
\\
&\qquad\qquad\qquad\qquad\qquad\qquad\qquad\qquad\qquad\qquad\qquad\qquad
+\frac{1}{2}\int_{0}^{1}\left(\frac{h'(t)}{\sigma(e^{h(t)})}\right)^{2}dt\Bigg\}.\label{I:rate:function}
\end{align}

(ii) For OTM variance put options, i.e. $V_{0}\eta^{2}(S_{0})>K$, we have
\begin{equation}
\lim_{T\rightarrow 0}T\log P(T)=-\mathcal{I}_\rho(S_{0},V_{0},K),
\end{equation}
where $\mathcal{I}_\rho(S_{0},V_{0},K)$ is defined in \eqref{I:rate:function}.
\end{theorem}

In Theorem~\ref{thm:OTM}, $\mathcal{I}_\rho(S_{0},V_{0},K)$
is the rate function from large deviations theory. 
It is written as the solution to a variational problem
optimizing over two functions $g,h$. 
In the ATM limit, i.e. $V_{0}\eta^{2}(S_{0})=K$, 
by letting $h'(t)\equiv 0$ and $g'(t)\equiv 0$ in \eqref{I:rate:function}, 
one gets $\mathcal{I}_\rho(S_{0},V_{0},K)=0$, which corresponds to the
law of large numbers limit.

The edge cases $\rho = \pm 1$ corresponding to perfectly correlated/anticorrelated asset and volatility processes, require a separate treatment 
due to the singular coefficient in the first term of \eqref{I:rate:function}, and will 
be covered in Proposition \ref{prop:pm:1}.

\begin{remark}\label{rmk:Heston}
Several models popular in financial practice, for example the Heston-type models with
$\sigma(V)=\sigma V^{-1/2}$ for some $\sigma>0$ in Eqn.~\eqref{eqn:V}, do not satisfy Assumptions~\ref{assump:bounded} and \ref{assump:LDP}.
Although Theorem~\ref{thm:OTM} is obtained under the Assumptions~\ref{assump:bounded} and \ref{assump:LDP}, one can see that these results hold also for this case as long as $\mathbb{Q}(\{(\log S_{Tt},\log V_{Tt}),0\leq t\leq 1\}\in\cdot)$ satisfies a sample-path large deviation principle (which is true for the Heston-type and CEV-type SDEs without Assumptions~\ref{assump:bounded} and \ref{assump:LDP} (see e.g. \cite{Baldi,Robertson2010})).
One can also easily check that $p$-th moments of $V_{t}$ are bounded on $[0,T]$ for any $p>1$ (see e.g. \cite{PWZ2024}) as in the proof of Theorem~\ref{thm:OTM}.
On the other hand, we do not pursue Heston-type models in our paper because they are disfavored by market data on options on realized variance. More precisely, they produce a down-sloping implied volatility smile for variance options, which is different from the observed trend (up-sloping smile); see e.g. \cite{Drimus2012}. 
\end{remark}

In the next section, we will analyze the variational problem \eqref{I:rate:function}. 
We will show that in the special case when $\rho=0$,
we can solve the variational problem \eqref{I:rate:function} in closed form,
and for general $\rho$, we will derive upper and lower bounds on the rate function. 
We give also explicit solutions in the limit $\rho=\pm 1$ of perfectly correlated and anti-correlated models. In the limit of a stochastic volatility model $\eta(\cdot )\equiv 1$, the rate function reduces to the rate function for Asian options in the local volatility model which was computed previously in \cite{PZAsian}. We confirm that the result obtained here reduces to the known result in the limit $\eta(\cdot)\equiv 1$. We also derive
an expansion of the rate function around the ATM point which can be used for practical applications.

\subsection{Variational problem}\label{sec:variational}

\subsubsection{Zero correlation}

In this section, we will show that for the special case $\rho=0$,
we can solve the variational problem \eqref{I:rate:function} in closed form.
When $\rho=0$, the underlying asset $S_{t}$ follows a stochastic volatility model
of the form:
\begin{align}
&\frac{dS_{t}}{S_{t}}=(r-q)dt+\eta(S_{t})\sqrt{V_{t}}dW_{t},
\\
&\frac{dV_{t}}{V_{t}}=\mu(V_{t})dt+\sigma(V_{t})dZ_{t},
\end{align}
where $W_{t}$ and $Z_{t}$ are two independent standard Brownian motions.
When $\rho=0$, we are able to solve
the variational problem \eqref{I:rate:function} in closed form.
The solution takes a simpler form for monotonic local volatility $\eta(x)$, but this assumption is not essential and can be relaxed. The general case of non-monotonic $\eta(x)$ can be treated by separating the range of $x$ into domains of monotonicity. For definiteness, we give next the result for a monotonic non-increasing $\eta(x)$.

\begin{proposition}\label{prop:variational}
Suppose the assumptions in Theorem~\ref{thm:OTM} hold.
Assume that the local volatility function $\eta(x)$ is non-increasing.
When $\rho=0$, the variational problem \eqref{I:rate:function}
has the following solution.

(i) For OTM variance call options, i.e. $V_{0}\eta^{2}(S_{0})<K$, we have
\begin{align}
&\mathcal{I}_0(S_{0},V_{0},K)
\nonumber
\\
&=\inf_{z}
\Bigg\{\frac{1}{K^{2}}
\left(\int_{S_{0}}^{G_{c}(z)}\frac{\eta(x)dx}{x\sqrt{2\left(\eta^{2}(G_{c}(z))-\eta^{2}(x)\right)}}\right)^{2}
\left(\eta^{2}(G_{c}(z))z-K\right)
+\mathcal{J}(V_{0},z)\Bigg\},\label{I:formula:call}
\end{align}
where $G_{c}(z)$ satisfies the equation
\begin{equation}\label{G:c:z}
\frac{\int_{S_{0}}^{G_{c}(z)}\frac{dx}{x\eta(x)\sqrt{\eta^{2}(G_{c}(z))-\eta^{2}(x)}}}
{\int_{S_{0}}^{G_{c}(z)}\frac{\eta(x)dx}{x\sqrt{\eta^{2}(G_{c}(z))-\eta^{2}(x)}}}
=\frac{z}{K},
\end{equation}
and
\begin{equation}\label{J:formula}
\mathcal{J}(V_{0},z):=\inf_{h(0)=\log V_{0},
\int_{0}^{1}e^{h(t)}dt=z}\frac{1}{2}\int_{0}^{1}\left(\frac{h'(t)}{\sigma(e^{h(t)})}\right)^{2}dt.
\end{equation}

(ii) For OTM variance put options, i.e. $V_{0}\eta^{2}(S_{0})>K$, 
we have
\begin{align}
&\mathcal{I}_0(S_{0},V_{0},K)
\nonumber
\\
&=\inf_{z}
\Bigg\{\frac{1}{K^{2}}
\left(\int_{S_{0}}^{G_{p}(z)}\frac{\eta(x)dx}{x\sqrt{2\left(\eta^{2}(x)-\eta^{2}(G_{p}(z))\right)}}\right)^{2}
\left(K-\eta^{2}(G_{p}(z))z\right)
+\mathcal{J}(V_{0},z)\Bigg\},\label{I:formula}
\end{align}
where $G_{p}(z)$ satisfies the equation
\begin{equation}\label{G:p:z}
\frac{\int_{S_{0}}^{G_{p}(z)}\frac{dx}{x\eta(x)\sqrt{\eta^{2}(x)-\eta^{2}(G_{p}(z))}}}
{\int_{S_{0}}^{G_{p}(z)}\frac{\eta(x)dx}{x\sqrt{\eta^{2}(x)-\eta^{2}(G_{p}(z))}}}
=\frac{z}{K},
\end{equation}
and $\mathcal{J}(V_{0},z)$ is defined in \eqref{J:formula}.
\end{proposition}

\begin{remark}
Note that when the options are ATM as $T\rightarrow 0$, i.e. when $V_{0}\eta^{2}(S_{0})=K$, the rate function
$\mathcal{I}_0(S_{0},V_{0},K)=0$, which is consistent with the law of large numbers.
Notice that when $V_{0}\eta^{2}(S_{0})=K$, 
if we take $z=V_{0}$, then $\mathcal{J}(V_{0},V_{0})=0$ since we can take $h(t)\equiv h(0)=\log V_{0}$ in \eqref{J:formula}, 
and moreover, with $z=V_{0}$, one can check that $G_{c}(z)=G_{p}(z)=S_{0}$ since
when $z=V_{0}$, we have $\frac{z}{K}=\frac{V_{0}}{K}=\frac{1}{\eta^{2}(S_{0})}$, 
and one can check that 
\begin{equation}
\lim_{G_{c}(z)\rightarrow S_{0}}\frac{\int_{S_{0}}^{G_{c}(z)}\frac{dx}{x\eta(x)\sqrt{\eta^{2}(G_{c}(z))-\eta^{2}(x)}}}
{\int_{S_{0}}^{G_{c}(z)}\frac{\eta(x)dx}{x\sqrt{\eta^{2}(G_{c}(z))-\eta^{2}(x)}}}
=\lim_{G_{p}(z)\rightarrow S_{0}}\frac{\int_{S_{0}}^{G_{p}(z)}\frac{dx}{x\eta(x)\sqrt{\eta^{2}(x)-\eta^{2}(G_{p}(z))}}}
{\int_{S_{0}}^{G_{p}(z)}\frac{\eta(x)dx}{x\sqrt{\eta^{2}(x)-\eta^{2}(G_{p}(z))}}}
=\frac{1}{\eta^{2}(S_{0})}.
\end{equation}
Hence, we conclude that when $V_{0}\eta^{2}(S_{0})=K$, the optimal $z=V_{0}$ in \eqref{I:formula:call} and \eqref{I:formula} and $G_{c}(z)=G_{p}(z)=\frac{1}{\eta^{2}(S_{0})}$.
\end{remark}

We illustrate the application of Proposition~\ref{prop:variational} in Section~\ref{sec:rho0} on a particular local volatility function $\eta(x)$ given by the Tanh model.
For this case the equations \eqref{G:c:z} and \eqref{G:p:z} have unique solutions for $G_c(z), G_p(z)$. In general they may have multiple solutions, and the infimum in \eqref{I:formula:call} and \eqref{I:formula} must be taken also over all solutions of these equations. 

\begin{remark}
The optimization problem
\begin{equation}
\mathcal{J}(V_{0},z)=\inf_{h(0)=\log V_{0},
\int_{0}^{1}e^{h(t)}dt=z}\frac{1}{2}\int_{0}^{1}\left(\frac{h'(t)}{\sigma(e^{h(t)})}\right)^{2}dt
\end{equation}
has already been solved in \cite{PZAsian}.
In particular, Proposition~8 in \cite{PZAsian}
showed that 
\begin{equation}
\mathcal{J}(V_{0},z)=
\begin{cases}
\frac12 F^{(-)}(f_1) G^{(-)}(f_1),
&\text{for $z>V_{0}$},
\\
\frac12 F^{(+)}(h_1) G^{(+)} (h_1),
&\text{for $z<V_{0}$},
\end{cases}
\end{equation}
where $f_1$ is the solution of the equation $e^{f_1} - z/V_0 = G^{(-)}(f_1)/F^{(-)}(f_1)$ and $h_1$ is the solution of the equation $z/V_0 - e^{-h_1} = G^{(+)}(h_1)/F^{(+)}(h_1)$. 
The functions $F^{(\pm)}(x), G^{(\pm)}(x)$ are defined as
\begin{align}
&G^{(-)}(x) := \int_0^{x} \frac{\sqrt{e^{x}-e^y}}{\sigma(V_0 e^y)} dy,\quad
F^{(-)}(x) := \int_0^{x} \frac{dy}{\sigma(V_0 e^y)\sqrt{e^{x}-e^y}} dy
\label{FG:minus:eqn}
\\
&G^{(+)}(x):=\int_0^x \frac{\sqrt{e^{-y}-e^{-x}}}{\sigma(V_0 e^{-y})} dy,\quad
F^{(+)}(x) := \int_0^{x} \frac{dy}{\sigma(V_0 e^{-y})\sqrt{e^{-y}-e^{-x}}} dy .
\label{FG:plus:eqn}
\end{align}
In particular, when $\sigma(\cdot)\equiv\sigma_{0}$, the solution simplifies and  is given by (Proposition 12 in \cite{PZAsian}) 
\begin{equation}
\mathcal{J}(V_{0},z)=
\begin{cases}
\frac{1}{\sigma_{0}^{2}}\left(\frac{1}{2}\beta^{2}-\beta\tanh\left(\frac{\beta}{2}\right)\right),
&\text{for $z>V_{0}$},
\\
\frac{2}{\sigma_{0}^{2}}\xi(\tan\xi-\xi),
&\text{for $z<V_{0}$},
\end{cases}
\end{equation}
where $\beta$ is the solution of the equation $\frac{1}{\beta}\sinh\beta=\frac{z}{V_{0}}$
for $z\geq V_{0}$ and
$\xi$ is the solution in the interval $[0,\frac{\pi}{2}]$ of the equation
$\frac{1}{2\xi}\sin(2\xi)=\frac{z}{V_{0}}$ for $z\leq V_{0}$.
\end{remark}

\begin{remark}\label{remark:upper:bound:by:J}
Note that for call options in \eqref{I:formula:call}, 
if we let $z$ such that $G_{c}(z)=S_{0}$,
then $\mathcal{I}_0(S_{0},V_{0},K)\leq\mathcal{J}(V_{0},z)$. 
Moreover, when $G_{c}(z)\rightarrow S_{0}$, it follows from \eqref{G:c:z}
that the left hand side of \eqref{G:c:z} converges to $\frac{1}{\eta^{2}(S_{0})}$ such that $z\rightarrow\frac{K}{\eta^{2}(S_{0})}$.
Thus, when $G_{c}(z)=S_{0}$, we have $z=\frac{K}{\eta^{2}(S_{0})}$.
Hence, we obtain the upper bound $\mathcal{I}_0(S_{0},V_{0},K)\leq\mathcal{J}\left(V_{0},\frac{K}{\eta^{2}(S_{0})}\right)$.
It is similar to check that a similar upper bound holds for the put options in \eqref{I:formula}.
\end{remark}

Proposition~\ref{prop:variational} solves the variational problem \eqref{I:rate:function}
when $\rho=0$ and obtains a simplified expression for $\mathcal{I}_0(S_{0},V_{0},K)$.
As a corollary from the proof of Proposition~\ref{prop:variational}, 
we are able to obtain the optimal $g$ and $h$ that solve the variational problem \eqref{I:rate:function} for the uncorrelated case $\rho=0$.

\begin{remark}\label{cor:variational}
Denote $g_{0},h_{0}$ the optimizers of the variational problem for the uncorrelated case $\rho=0$. They are given explicitly as follows.

(i) For OTM variance call options, i.e. $V_{0}\eta^{2}(S_{0})<K$, 
then the optimal $g_{0}$ is given by
\begin{equation}
\int_{S_{0}}^{e^{g_{0}(t)}}\frac{dx}{x\eta(x)\sqrt{2\left(\eta^{2}(G_{c}(z_{c}))-\eta^{2}(x)\right)}}
=\frac{\int_{S_{0}}^{G_{c}(z_{c})}\frac{\eta(x)dx}{x\sqrt{2(\eta^{2}(G_{c}(z_{c}))-\eta^{2}(x))}}}{K}\int_{0}^{t}e^{h_{0}(s)}ds,
\end{equation}
where the optimal $h_{0}(t)=\log V_{0}+f_{0}(t;z_{c})$ with
\begin{equation}
\begin{cases}
\int_{0}^{f_{0}(t;z)}\frac{dy}{\sigma(V_{0}e^{y})\sqrt{e^{\alpha_{-}}-e^{y}}}=F^{(-)}(\alpha_{-}(z))t &\text{when $z>V_{0}$}, 
\\
\int_{0}^{f_{0}(t;z)}\frac{dy}{\sigma(V_{0}e^{y})\sqrt{e^{y}-e^{-\alpha_{+}}}}=-F^{(+)}(\alpha_{+}(z))t &\text{when $z<V_{0}$}, 
\end{cases}
\end{equation}
for any $0\leq t\leq 1$, where $\alpha_{+}=\alpha_{+}(z)$ is the solution of the equation
\begin{equation}
\frac{z}{V_0} - e^{-\alpha_{+}} = \frac{G^{(+)}(\alpha_{+})}{F^{(+)}(\alpha_{+})},
\end{equation}
with
\begin{align}
G^{(+)}(\alpha_{+}) = \int_{0}^{\alpha_{+}} \frac{ \sqrt{e^{-y} - e^{-\alpha_{+}}}}{\sigma( V_0 e^{-y})} dy, \quad
F^{(+)}(\alpha_{+}) = \int_{0}^{\alpha_{+}} \frac{1}{\sigma( V_0 e^{-y})} 
\frac{1}{\sqrt{e^{-y} - e^{-\alpha_{+}}}}dy \,,
\end{align}
and $\alpha_{-}=\alpha_{-}(z)$ is the solution of the equation
\begin{equation}
e^{\alpha_{-}} - \frac{z}{V_0} = \frac{G^{(-)}(\alpha_{-})}{F^{(-)}(\alpha_{-})},
\end{equation}
with
\begin{align}
G^{(-)}(\alpha_{-})
=\int_{0}^{\alpha_{-}} \frac{\sqrt{e^{\alpha_{-}} - e^{y}}}{\sigma(V_0 e^{y})} dy,
\quad
F^{(-)}(\alpha_{-})=\int_{0}^{\alpha_{-}} \frac{1}{\sigma(V_0 e^{y})} 
\frac{1}{\sqrt{e^{\alpha_{-}} -  e^{y}}}dy \,.
\end{align}

(ii) For OTM variance put options, i.e. $V_{0}\eta^{2}(S_{0})>K$,
then the optimal $g_{0}$ is given by
\begin{equation}
\int_{S_{0}}^{e^{g_{0}(t)}}\frac{dx}{x\eta(x)\sqrt{2\left(\eta^{2}(x)-\eta^{2}(G_{p}(z_{p}))\right)}}
=\frac{\int_{S_{0}}^{G_{p}(z_{p})}\frac{\eta(x)dx}{x\sqrt{2(\eta^{2}(x)-\eta^{2}(G_{p}(z_{p})))}}}{K}\int_{0}^{t}e^{h_{0}(s)}ds,
\end{equation}
where the optimal $h_{0}(t)=\log V_{0}+f_{0}(t;z_{p})$, with $f_{0}(t;z)$ being defined in (i).
\end{remark}


\subsubsection{Non-zero correlation}

In this section, we discuss the general $\rho\neq 0$ case.
We first recall the variational problem for the rate function, see \eqref{I:rate:function}
\begin{align}
\mathcal{I}_{\rho}(S_{0},V_{0},K)
=\inf_{\substack{g(0)=\log S_{0},h(0)=\log V_{0}\\
\int_{0}^{1}e^{h(t)}\eta^{2}(e^{g(t)})dt=K}}\Lambda_{\rho}[g,h],\label{I:rho:1}
\end{align}
where
\begin{align}
\Lambda_{\rho}[g,h]:=\frac{1}{2(1-\rho^{2})}\int_{0}^{1}\left(\frac{g'(t)}{\eta(e^{g(t)})\sqrt{e^{h(t)}}}-\frac{\rho h'(t)}{\sigma(e^{h(t)})}\right)^{2}dt
+\frac{1}{2}\int_{0}^{1}\left(\frac{h'(t)}{\sigma(e^{h(t)})}\right)^{2}dt \,.\label{I:rho:2}
\end{align}

The optimal $g$ and $h$ for the variational problem \eqref{I:rho:1} can be determined implicitly via the Euler-Lagrange equations. First, let us define:
\begin{align}\label{I:rho:lambda}
\Lambda_{\rho}[g,h;\lambda]&:=\frac{1}{2(1-\rho^{2})}\int_{0}^{1}\left(\frac{g'(t)}{\eta(e^{g(t)})\sqrt{e^{h(t)}}}-\frac{\rho h'(t)}{\sigma(e^{h(t)})}\right)^{2}dt
+\frac{1}{2}\int_{0}^{1}\left(\frac{h'(t)}{\sigma(e^{h(t)})}\right)^{2}dt
\nonumber
\\
&\qquad
+\lambda\int_{0}^{1}e^{h(t)}\eta^{2}(e^{g(t)})dt
\nonumber
\\
&=\frac{1}{2(1-\rho^{2})}\int_{0}^{1}\left(\frac{g'(t)}{\eta(e^{g(t)})\sqrt{e^{h(t)}}}\right)^{2}dt
+\frac{1}{2(1-\rho^{2})}\int_{0}^{1}\left(\frac{h'(t)}{\sigma(e^{h(t)})}\right)^{2}dt
\nonumber
\\
&\qquad
-\frac{\rho}{1-\rho^{2}}\int_{0}^{1}\frac{g'(t)}{\eta(e^{g(t)})\sqrt{e^{h(t)}}}\frac{h'(t)}{\sigma(e^{h(t)})}dt
+\lambda\int_{0}^{1}e^{h(t)}\eta^{2}(e^{g(t)})dt,
\end{align}
where $\lambda$ is the Lagrange multiplier.
The optimal $g$ and $h$ for the variational problem \eqref{I:rho:1} satisfy the Euler-Lagrange equations for the functional $\Lambda_\rho$. They are given by
$\frac{\delta \Lambda_{\rho}}{\delta g}=\frac{d}{dt}\left(\frac{\delta \Lambda_{\rho}}{\delta g'}\right)$
and $\frac{\delta \Lambda_{\rho}}{\delta h}=\frac{d}{dt}\left(\frac{\delta \Lambda_{\rho}}{\delta h'}\right)$
which leads to
\begin{align}
&\frac{d}{dt}\left(\frac{1}{1-\rho^{2}}\frac{g'(t)}{\eta^{2}(e^{g(t)})e^{h(t)}}
-\frac{\rho}{1-\rho^{2}}\frac{h'(t)}{\eta(e^{g(t)})\sqrt{e^{h(t)}}\sigma(e^{h(t)})}\right)
\nonumber
\\
&=-\frac{1}{1-\rho^{2}}\frac{(g'(t))^{2}\eta'(e^{g(t)})e^{g(t)}}{\eta^{3}(e^{g(t)})e^{h(t)}}
+\frac{\rho}{1-\rho^{2}}\frac{g'(t)h'(t)\eta'(e^{g(t)})e^{g(t)}}{\eta^{2}(e^{g(t)})\sqrt{e^{h(t)}}\sigma(e^{h(t)})}
\nonumber
\\
&\qquad\qquad\qquad
+2\lambda e^{h(t)}\eta(e^{g(t)})\eta'(e^{g(t)})e^{g(t)},\label{EL:1}
\end{align}
and
\begin{align}
&\frac{d}{dt}\left(\frac{1}{1-\rho^{2}}\frac{h'(t)}{\sigma^{2}(e^{h(t)})}
-\frac{\rho}{1-\rho^{2}}\frac{g'(t)}{\eta(e^{g(t)})\sqrt{e^{h(t)}}\sigma(e^{h(t)})}\right)
\nonumber
\\
&=-\frac{1}{2(1-\rho^{2})}\frac{(g'(t))^{2}}{\eta^{2}(e^{g(t)})e^{h(t)}}
-\frac{1}{1-\rho^{2}}\frac{(h'(t))^{2}\sigma'(e^{h(t)})e^{h(t)}}{\sigma^{3}(e^{h(t)})}
\nonumber
\\
&\qquad
+\frac{\rho}{1-\rho^{2}}\frac{g'(t)}{\eta(e^{g(t)})\sqrt{e^{h(t)}}}\frac{h'(t)\sigma'(e^{h(t)})e^{h(t)}}{\sigma^{2}(e^{h(t)})}
+\frac{\rho}{2(1-\rho^{2})}\frac{g'(t)}{\eta(e^{g(t)})\sqrt{e^{h(t)}}}\frac{h'(t)}{\sigma(e^{h(t)})}
\nonumber
\\  
&\qquad\qquad\qquad
+\lambda e^{h(t)}\eta^{2}(e^{g(t)}),\label{EL:2}
\end{align}
with the constraints $g(0)=\log S_{0}$, $h(0)=\log V_{0}$ and $\int_{0}^{1}e^{h(t)}\eta^{2}(e^{g(t)})dt=K$.
The transversality condition gives $g'(1)=h'(1)=0$.

Even though it seems impossible to solve the 
Euler-Lagrange equations~\eqref{EL:1}-\eqref{EL:2} (and hence the variational problem for $\mathcal{I}_{\rho}(S_{0},V_{0},K)$)
in closed form,
we can obtain the following lower and upper bounds for $\mathcal{I}_{\rho}(S_{0},V_{0},K)$.

\begin{proposition}\label{prop:bounds}
For any $\rho\in(-1,1)$, we have
\begin{equation}
\frac{1}{1+|\rho|}\mathcal{I}_{0}(S_{0},V_{0},K)
\leq\mathcal{I}_{\rho}(S_{0},V_{0},K)
\leq\frac{1}{1-|\rho|}\mathcal{I}_{0}(S_{0},V_{0},K),
\end{equation}
where $\mathcal{I}_{0}(S_{0},V_{0},K)$ is computed in closed form in Proposition~\ref{prop:variational}.
\end{proposition}

The bounds in Proposition~\ref{prop:bounds} can be used to obtain bounds
for the asymptotics of the OTM variance options.
We also notice that Proposition~\ref{prop:bounds} works well
when $|\rho|$ is small.
Indeed, when $|\rho|$ is small, we expect
that the optimal $g,h$ for the variational problem \eqref{I:rho:1}
should be close to $g_{0},h_{0}$ which are the optimal solutions
for the variational problem \eqref{I:rho:1} when $\rho=0$ that
can be solved analytically; see Corollary~\ref{cor:variational}.
This helps us establish the following upper bound
for $\mathcal{I}_{\rho}(S_{0},V_{0},K)$ which we expect
to work well when $|\rho|$ is small.

\begin{proposition}\label{prop:bounds:2}
For any $\rho\in(-1,1)$, we have
\begin{align}
\mathcal{I}_{\rho}(S_{0},V_{0},K)
\leq\Lambda_{\rho}[g_{0},h_{0}],
\end{align}
where $\Lambda_{\rho}[\cdot,\cdot]$ is defined in \eqref{I:rho:2}
and $g_{0},h_{0}$ are the optimal solutions for the variational problem \eqref{I:rho:1} 
when $\rho=0$.
\end{proposition}

Next, we provide another upper bound for $\mathcal{I}_{0}(S_{0},V_{0},K)$, 
which is an extension to the upper bound in Remark~\ref{remark:upper:bound:by:J} when $\rho=0$.

\begin{proposition}\label{prop:bounds:3}
For any $\rho\in(-1,1)$, we have 
\begin{align}
\mathcal{I}_{\rho}(S_{0},V_{0},K)
\leq\frac{1}{1-\rho^{2}}\mathcal{J}\left(V_{0},\frac{K}{\eta^{2}(S_{0})}\right),
\end{align}
where $\mathcal{J}(\cdot,\cdot)$ is defined in \eqref{J:formula}.
\end{proposition}


Although Proposition~\ref{prop:bounds} works well
when $|\rho|$ is small, when $|\rho|\rightarrow 1$, the upper bound in Proposition~\ref{prop:bounds} becomes trivial. 
Next, we will analyze the $|\rho|\rightarrow 1$ case in detail.

\subsection{Perfectly correlated and anti-correlated cases}\label{sec:perfect}

In this section, we consider the case of perfect correlation $\rho=+1$ and 
perfectly anti-correlated asset price and volatility $\rho = -1$. We show that for
these cases 
the variational problem for $\mathcal{I}_{\rho}(S_{0},V_{0},K)$
given in Theorem~\ref{thm:OTM} can be further analyzed.
We have the following result.

\begin{proposition}\label{prop:pm:1}
Denote $\mathcal{I}_{\pm 1}(S_{0},V_{0},K)$ the rate function for the perfectly correlated/anti-correlated case $\rho=\pm 1$. Then we have
\begin{equation}\label{pm:variational}
\mathcal{I}_{\pm 1}(S_{0},V_{0},K)
=\inf_{\substack{h(0)=\log V_{0}\\
\int_{0}^{1}e^{h(t)}\eta^{2}(\mathcal{F}_{\pm}(e^{h(t)}))dt=K}}\frac{1}{2}\int_{0}^{1}\left(\frac{h'(t)}{\sigma(e^{h(t)})}\right)^{2}dt,
\end{equation}
where $\mathcal{F}_{\pm}(\cdot)$ is defined as:
\begin{equation}\label{pm:F}
\int_{S_{0}}^{\mathcal{F}_{\pm}(x)}\frac{dy}{y\eta(y)}=\int_{V_{0}}^{x}\frac{\pm dy}{\sqrt{y}\sigma(y)},\qquad\text{for any $x>0$}.
\end{equation}
\end{proposition}

\begin{remark}\label{remark:pm:1}
For the special case $\eta(\cdot)\equiv\eta_{0}$, 
the optimization problem
\begin{equation}
\mathcal{J}(V_{0},K/\eta_{0}^{2})=\inf_{h(0)=\log V_{0},
\int_{0}^{1}e^{h(t)}dt=K/\eta_{0}^{2}}\frac{1}{2}\int_{0}^{1}\left(\frac{h'(t)}{\sigma(e^{h(t)})}\right)^{2}dt
\end{equation}
has already been solved in \cite{PZAsian}.
The solution is given in Proposition~8 in \cite{PZAsian} and reads
\begin{equation}
\mathcal{J}(V_{0},K/\eta_{0}^{2})=
\begin{cases}
\frac12 F^{(-)}(f_1) G^{(-)}(f_1),
&\text{for $K/\eta_{0}^{2}>V_{0}$},
\\
\frac12 F^{(+)}(h_1) G^{(+)} (h_1),
&\text{for $K/\eta_{0}^{2}<V_{0}$},
\end{cases}
\end{equation}
where $f_1$ is the solution of the equation $e^{f_1} - K/(\eta_{0}^{2}V_0) = G^{(-)}(f_1)/F^{(-)}(f_1)$ and $h_1$ is the solution of the equation $K/(\eta_{0}^{2}V_0) - e^{-h_1} = G^{(+)}(h_1)/F^{(+)}(h_1)$, 
where $\mathcal{F}^{(\mp)}(\cdot)$ and $\mathcal{G}^{(\mp)}(\cdot)$ are defined in \eqref{FG:minus:eqn}-\eqref{FG:plus:eqn}.

In the particular limit when $\sigma(\cdot)\equiv\sigma_{0}$, this simplifies further and can be expressed in closed form through $J_{\mathrm{BS}}(k)$, the rate function for Asian options in the Black-Scholes model, given by
(see Proposition~12 in \cite{PZAsian})
\begin{equation}\label{JBS:def}
\mathcal{J}(V_{0},K/\eta_{0}^{2})= \frac{1}{\sigma_0^2}J_{\mathrm{BS}}(K/(\eta_0^2 V_0)) :=
\begin{cases}
\frac{1}{\sigma_{0}^{2}}\left(\frac{1}{2}\beta^{2}-\beta\tanh\left(\frac{\beta}{2}\right)\right),
&\text{for $K>\eta_{0}^{2}V_{0}$},
\\
\frac{2}{\sigma_{0}^{2}}\xi(\tan\xi-\xi),
&\text{for $K<\eta_{0}^{2}V_{0}$},
\end{cases}
\end{equation}
where $\beta$ is the solution of the equation $\frac{1}{\beta}\sinh\beta=\frac{K}{\eta_{0}^{2}V_{0}}$
for $K\geq\eta_{0}^{2}V_{0}$ and
$\xi$ is the solution in the interval $[0,\frac{\pi}{2}]$ of the equation
$\frac{1}{2\xi}\sin(2\xi)=\frac{K}{\eta_{0}^{2}V_{0}}$ for $K\leq\eta_{0}^{2}V_{0}$.
\end{remark}

\begin{remark}\label{remark:pm:general}
More generally, the variational problem \eqref{pm:variational}
can be solved analytically. By a change of variable it can be reformulated 
as the variational problem for the asymptotics of Asian options
in local volatility models which was solved in \cite{PZAsian}.
To see this, let us first define $\mathcal{G}_{\pm}(x):=x\eta^{2}(\mathcal{F}_{\pm}(x))$
for any $x>0$ and assume that  the inverse function $\mathcal{G}^{-1}_{\pm}(\cdot)$ exists. 
Next, let $g(t)=\log\mathcal{G}_{\pm}(e^{h(t)})$ in \eqref{pm:variational}, 
such that $h(t)=\log\mathcal{G}_{\pm}^{-1}(e^{g(t)})$.
Then, we can rewrite \eqref{pm:variational} as
\begin{equation}\label{pm:variational:2}
\mathcal{I}_{\pm 1}(S_{0},V_{0},K)
=\inf_{\substack{g(0)=\log\mathcal{G}_{\pm}(V_{0})\\
\int_{0}^{1}e^{g(t)}dt=K}}\frac{1}{2}\int_{0}^{1}\left(\frac{g'(t)}{\hat{\sigma}(e^{g(t)})}\right)^{2}dt\,,
\end{equation}
where $\hat{\sigma}(S):=\frac{1}{S}\mathcal{G}_{\pm}^{-1}(S)\mathcal{G}'_{\pm}(\mathcal{G}^{-1}_{\pm}(S))\sigma(\mathcal{G}^{-1}_{\pm}(S))$ for any $S>0$.
Then, \eqref{pm:variational:2} is exactly the rate function for Asian option
for local volatility models with the local volatility $\hat{\sigma}(\cdot)$, the spot asset price $\mathcal{G}_{\pm}(V_{0})$ and the strike price $K$; see \cite{PZAsian}.
The variational problem \eqref{pm:variational:2} has been solved analytically in \cite{PZAsian}.
\end{remark}

Note that Proposition~\ref{prop:pm:1}
concerns the case $\rho=\pm 1$. 
When $\rho$ is close to $\pm 1$, 
the factor $\frac{1}{2(1-\rho^{2})}$ in \eqref{I:rho:2}
is large, and we expect that
the choice of $g,h$ to make the first term in \eqref{I:rho:2}
zero becomes a good choice since otherwise
the first term in \eqref{I:rho:2} would become large when
$\rho$ is close to $\pm 1$. 
Using this intuition, we obtain the following result,
that provides an upper bound for $\mathcal{I}_{\rho}(S_{0},V_{0},K)$
in \eqref{I:rho:1} and we expect that this provides a good
approximation when $\rho$ is close to $\pm 1$.

\begin{proposition}\label{prop:rho}
For any $\rho\in(-1,1)$,
\begin{equation}
\mathcal{I}_{\rho}(S_{0},V_{0},K)
\leq\inf_{\substack{h(0)=\log V_{0}\\
\int_{0}^{1}e^{h(t)}\eta^{2}(\mathcal{F}_{\rho}(e^{h(t)}))dt=K}}\frac{1}{2}\int_{0}^{1}\left(\frac{h'(t)}{\sigma(e^{h(t)})}\right)^{2}dt,
\end{equation}
where $\mathcal{F}_{\rho}(\cdot)$ is defined as:
\begin{equation}
\int_{S_{0}}^{\mathcal{F}_{\rho}(x)}\frac{dy}{y\eta(y)}=\int_{V_{0}}^{x}\frac{\rho dy}{\sqrt{y}\sigma(y)},\qquad\text{for any $x>0$}.
\end{equation}
\end{proposition}

\begin{remark}
For the special case $\eta(\cdot)\equiv\eta_{0}$, 
the upper bound in Proposition~\ref{prop:rho}
can be solved in closed form as discussed in Remark~\ref{remark:pm:1}.
More generally, the upper bound in Proposition~\ref{prop:rho}
can be reformulated and solved as discussed in Remark~\ref{remark:pm:general}.
\end{remark}

The previous simplification of the variational problem still requires solving
a variational problem of a single function or can be reduced further
to a finite-dimensional optimization problem. 

\subsection{Expansion of the rate function around the ATM point}\label{sec:expansion:around:ATM}

In this section we consider the expansion of the rate function around the ATM point.
Our approach is inspired by the recent work \cite{PZAsianLSV} that conducts an expansion of the rate function around the ATM point for Asian options under the local-stochastic volatility model \eqref{eqn:S}-\eqref{eqn:V} to solve a two-dimensional variational problem. We mention also the expansion 
used in \cite{Bayer2017,Friz2022} to solve a variational problem for the rate function for a stochastic volatility model with rough volatility.
The expansion approach has the advantage of providing explicit formulas that are easier to use in practice than a direct numerical solution of the variational problem for the rate function.

Recall that the ATM case corresponds to strike $K=V_{0}\eta^{2}(S_{0})$. 
Thus, we are seeking the solution of the variational problem in an expansion around the ATM point using the log-moneyness  $x:=\log\left(\frac{K}{V_{0}\eta^{2}(S_{0})}\right)$ as a small parameter.

The result is formulated in terms of the coefficients in the expansion of the local volatility function around $S_0$
\begin{equation}\label{eta:expansion}
\eta(S) = \eta_0 + \eta_1 \log \frac{S}{S_0} + \eta_2 \log^2 \frac{S}{S_0} + O\left(\log^3(S/S_0)\right),
\end{equation}
and analogous for the expansion of the volatility-of-volatility function around $V_0$
\begin{equation}\label{sigma:expansion}
\sigma(V) = \sigma_0 + \sigma_1 \log \frac{V}{V_0} + \sigma_2 \log^2 \frac{V}{V_0} 
+ O\left(\log^3(V/V_0)\right) \,.
\end{equation}
More explicitly, $\eta_0 = \eta(S_0), \eta_1=S_0 \eta'(S_0)$ and $\sigma_0=\sigma(V_0), \sigma_1 = V_0 \sigma'(V_0)$.
We have the following expansion of the rate function in powers of $x:=\log\left(\frac{K}{V_{0}\eta_0^{2}}\right)$.

\begin{proposition}\label{prop:first:order}
Suppose that $\sigma(\cdot), \eta(\cdot )$ are twice continuously differentiable such that the expansions \eqref{sigma:expansion} and \eqref{eta:expansion} are valid. 
Then we have the following expansion of the rate function in powers of the log-moneyness $x:=\log\left(\frac{K}{V_{0}\eta^{2}(S_{0})}\right)$:
\begin{align}\label{JAexp}
& \mathcal{I}_{\rho}(S_{0},V_{0},V_{0}\eta^{2}(S_{0})e^{x})
=\frac{3}{2\left(\sigma_{0}^{2} + 4\rho\sigma_{0} \sqrt{V_0} \eta_{1}+4\eta_{1}^{2} V_0\right)}x^{2} \\
& \quad
- \frac{3}{10(\sigma_{0}^{2} + 4\rho\sigma_{0} \sqrt{V_0} \eta_{1}+4\eta_{1}^{2} V_0)^3} \left( \sigma_0^4 + 2\sigma_0^3\left( 3\sigma_1 + 7\eta_1 \rho \sqrt{V_0}\right) 
+ \beta_2 \sigma_0^2 + \beta_1 \sigma_0 + \beta_0 \right) x^3 \nonumber \\
& \qquad
+O(x^{4}), \nonumber
\end{align}
as $x\rightarrow 0$, with
\begin{align}\label{beta0}
\beta_0 &= 16\eta_1^2 V_0^2 \left(\eta_1^2 + 6\eta_0 \eta_2\right), \\
\beta_1 &= 8\eta_1 \rho V_0 \left(3\eta_1 \rho \sigma_1 + 7\eta_1^2 \sqrt{V_0} + 12 \eta_0 \eta_2 \sqrt{V_0} \right), \\
\label{beta2}
\beta_2 &= 4\left(6\eta_1 \rho \sigma_1 \sqrt{V_0} + 6\eta_0 \eta_2 \rho^2 V_0 + \eta_1^2 \left(5+7\rho^2\right) V_0\right) \,.
\end{align}


The optimal $g,h$ that solve the variational problem \eqref{I:rho:1} admit 
the expansion
\begin{equation}\label{g:h:expansion:first}
g(t)=g_{0}(t)+xg_{1}(t)+O(x^{2}),
\qquad
h(t)=h_{0}(t)+xh_{1}(t)+O(x^{2}),
\end{equation}
as $x\rightarrow 0$, where $g_{0}(t)\equiv\log S_{0}$, $h_{0}(t)\equiv\log V_{0}$ and
\begin{align}
&g_{1}(t)=\frac{3}{2} \frac{(2 \eta_1 \sqrt{V_0} + \rho \sigma_0) \sqrt{V_0}}{\sigma_{0}^{2} + 4\rho\sigma_{0} \sqrt{V_0} \eta_{1}+4\eta_{1}^{2} V_0}(2t-t^{2}),
\\
&h_{1}(t)=\frac{3}{2}\frac{\sigma_0 (\sigma_0 + 2\eta_1 \rho \sqrt{V_0})}{\sigma_{0}^{2} + 4\rho\sigma_{0} \sqrt{V_0} \eta_{1}+4\eta_{1}^{2} V_0}
(2t-t^{2}).
\end{align}
The functions $g_2(t), h_2(t)$ are given in Appendix \ref{sec:g2:h2}.
\end{proposition}

In the limit $\eta(x)\equiv 1$ the rate function for variance options reduces to the rate function for Asian options in a local volatility model with volatility $\sigma(v)$ \cite{PZAsian}. The expansion of this rate function in powers of log-strike was computed for the latter case in Corollary~16 of \cite{PZAsian}. We can check that the two results agree indeed. 
By taking $\eta_0=1, \eta_2=\eta_3 = 0$ in (\ref{JAexp}), we get
\begin{align}\label{JAexpSV}
 \mathcal{I}_{\rho}(S_{0},V_{0},V_{0}\eta^{2}(S_{0})e^{x})\Big|_{\eta(x)\equiv 1}
=\frac{3}{2\sigma_0^2 }x^{2} - \frac{3}{10\sigma_{0}^{3} } \left( \sigma_0 + 6 \sigma_1   \right) x^3  +O(x^{4}), 
\end{align}
which agrees indeed with the first two terms in equation~(40) of \cite{PZAsian}.

\subsection{ATM case}\label{sec:ATM}

We give in this section the leading short maturity asymptotics for ATM variance options. 
Recall that in the short maturity limit $T\to 0$ the ATM case corresponds to variance options with strike $K=V_{0}\eta^{2}(S_{0})$.
Unlike the OTM case, the $T$ scaling for the ATM variance options 
will be seen to be different.
In probabilistic language, the ATM regime
corresponds to fluctuations associated with the central limit theorem
regime, which is in contrast to the large deviations regime that dominates the short maturity asymptotics for the OTM case.
We have the following result.

\begin{theorem}\label{thm:ATM}
Suppose that Assumptions~\ref{assump:bounded}, \ref{assump:lip} hold.
Further assume that $\eta(\cdot)$ is twice differentiable with $\sup_{x\in\mathbb{R}^{+}}|(\eta^{2})''(x)|<\infty$
and there exists some $C'\in(0,\infty)$ 
such that $\max_{0\leq t\leq T}\mathbb{E}[(S_{t})^{4}]\leq C'$
for any sufficiently small $T>0$.

For ATM variance call and put options, i.e. $K=V_{0}\eta^{2}(S_{0})$, 
we have
\begin{align}
&\lim_{T\rightarrow 0}\frac{C(T)}{\sqrt{T}}
=\lim_{T\rightarrow 0}\frac{P(T)}{\sqrt{T}}
\nonumber
\\
&=\frac{1}{\sqrt{6\pi}}(\eta^2(S_0) V_0) 
\sqrt{4V_{0}(S_0 \eta'(S_{0}))^{2}
+ \sigma^{2}(V_{0})
+4\rho S_0 \eta'(S_{0}) \sqrt{V_{0}} \sigma(V_{0})}.
\end{align}
\end{theorem}


Note that the small-maturity asymptotics of the ATM option prices is of the order $O(\sqrt{T})$, 
in contrast with that of the OTM options which are exponentially suppressed in the order of $O(e^{-1/T})$.
This is similar to the behavior obtained for Asian options in \cite{PZAsian}.

\section{Applications and Numerical Tests}
\label{sec:appl}

We present in this section the application of our asymptotic results for the pricing of 
variance options. Following the same approach as that used for Asian options in 
\cite{PZAsian}, the prices of these options are represented in Black-Scholes form as
\begin{align}\label{CPV}
C(K,T) &= e^{-rT} [ F_V(T) N(d_1) - K N(d_2)]\,, \\
P(K,T) &= e^{-rT} [ - F_V(T) N(- d_1) + K N(-d_2)]\,, \nonumber
\end{align}
where $F_V(T)$ is the varswap fair strike defined in (\ref{FVT}), $N(\cdot)$ is the cumulative distribution function of a standard Gaussian random variable with mean $0$ and variance $1$ and $d_{1,2}
= \frac{1}{\Sigma_V \sqrt{T}} ( \log(F_V(T)/K) \pm \frac12 \Sigma_V^2 T)$, with
$\Sigma_V(K,T)$ being the \textit{implied volatility of a variance option}. This volatility
is defined such that the price of the variance option is equal to the Black-Scholes price given above in (\ref{CPV}).

The short-maturity limit in Theorem~\ref{thm:OTM} implies a prediction for the implied 
volatility of the variance options
\begin{equation}\label{SigVdef}
\lim_{T\to 0} \Sigma_V^2(K,T) := \Sigma_V^2(K) = 
\frac{\log^2(K/F_V(T))}{2 \mathcal{I}_\rho(S_0,V_0,K)} \,.
\end{equation}

Using the expansion of the rate function around the ATM point in Proposition~\ref{prop:first:order} 
yields an expansion of the asymptotic implied volatility $\Sigma_V(K)$ in powers of log-strike: 
\begin{equation}\label{SigVapp}
\Sigma_V(K) = \Sigma_{V,\mathrm{ATM}} + s_V x + O(x^2)\,,
\end{equation}
where 
\begin{equation}\label{SigVATM}
\Sigma_{V,\mathrm{ATM}} = \frac{1}{\sqrt3} \sqrt{\sigma_0^2 + 4\rho \sigma_0 \sqrt{V_0} \eta_1 + 4\eta_1^2 V_0}\,,
\end{equation}
and 
\begin{align}
s_V = \frac{\sigma_0^4 + 2\sigma_0^3\left(3\sigma_1 + 7\eta_1 \rho \sqrt{V_0}\right) + \beta_2 \sigma_0^2 + \beta_1 \sigma_0 + \beta_0}{10 \sqrt3 (\sigma_0^2 + 4\rho \sigma_0 \sqrt{V_0} \eta_1 + 4\eta_1^2 V_0)^{3/2}}\,,
\end{align}
where $\beta_{0,1,2}$ are given above in (\ref{beta0}) - (\ref{beta2}).

We will use the approximation (\ref{SigVapp}) for the numerical tests of the asymptotic expansion in the next section.


For the numerical tests, we will assume the local-stochastic volatility model with log-normal volatility
\begin{equation}
dS_t/S_t = \eta(S_t) \sqrt{V_t} dB_t \,,\qquad dV_t/V_t= \sigma dZ_t\,,
\end{equation}
where $(B_t, Z_t)$ are two standard Brownian motions correlated with correlation $\rho$, 
and the local volatility function is given by
\begin{equation}\label{loc:vol:eta:S} 
\eta(S) = f_0 + f_1 \tanh\left(\log (S/S_0) - x_0\right) \,.
\end{equation}

This is the so-called Tanh model which was used in Forde and Jacquier (2011) \cite{Forde2011}
 to test the predictions of their asymptotic results for the uncorrelated local-stochastic volatility model. This model was also used in \cite{PWZ2024} for numerical tests of the short maturity asymptotics for VIX options.

The local volatility function \eqref{loc:vol:eta:S} is expanded in powers of the log-asset $\log(S/S_0)$ as
\begin{equation}
\eta(S) = \eta_0 + \eta_1 \log\frac{S}{S_0} +  \eta_2 \log^2\frac{S}{S_0} + \cdots\,,
\end{equation}
with 
\begin{eqnarray}
\eta_0 = f_0 - f_1 \tanh x_0\,, \quad \eta_1 = \frac{f_1}{\cosh^2 x_0}\,,  \quad
\eta_2 = \frac{f_1}{\cosh^2 x_0} \tanh x_0 \,.
\end{eqnarray}



\subsection{The uncorrelated case}
\label{sec:rho0}

In this section we present a test for the solution of the variational problem \eqref{I:rate:function} for the rate function which does not rely on an expansion in log-moneyness. For the uncorrelated case $\rho=0$ an exact solution for the variational problem was given in Proposition~\ref{prop:variational}. We present a test of the accuracy of the rate function obtained by expanding in log-moneyness in Proposition~\ref{prop:rho} by comparing with the exact solution following from Proposition~\ref{prop:variational}.

For this test, we take the constant vol-of-vol $\sigma(v) \equiv \sigma$ and $\mu(v)\equiv 0$.
The rate function $\mathcal{J}(V_0, K/\eta_0^2)$ corresponding to $\sigma(v)\equiv\sigma$ is the same as the rate function for Asian options in the Black-Scholes model and is known exactly from Proposition~12 in \cite{PZAsian}.
The result of Proposition~\ref{prop:variational} can be reformulated in a simpler form as follows.

i) For the OTM variance call options, we have
\begin{equation}
\mathcal{I}_0(S_0,V_0,K) = \inf_y \left( \frac{1}{K} F_c(y) +
\frac{1}{\sigma^2} J_{\mathrm{BS}}(y K/V_0) \right)\,,
\end{equation}
where $J_{\mathrm{BS}}(k)$ is defined in \eqref{JBS:def} which is the rate function for Asian options in the Black-Scholes model and is given in closed form in Proposition~12 of \cite{PZAsian}, and
\begin{equation}
F_c(y) := \frac12 \left(I_1^c(g_c(y))\right)^2 
\left[ \left(\eta(g_c(y))\right)^2 y -1 \right]\,, 
\end{equation}
where $g_c(y)$ is the solution of the equation
$\frac{I_2^c(g_c)}{I_1^c(g_c)} = y$
and
\begin{align}
I_1^c(g_c) := \int_{S_0}^{g_c} \frac{\eta(x) dx}{x \sqrt{\eta^2(g_c) - \eta^2(x)} }\,,\quad
I_2^c(g_c) := \int_{S_0}^{g_c} \frac{dx}{x\eta(x) \sqrt{\eta^2(g_c) - \eta^2(x)} }\,.
\end{align}

ii) For the OTM variance put options, we have
\begin{equation}
\mathcal{I}_0(S_0,V_0,K) = \inf_y \left( \frac{1}{K} F_p(y) +
\frac{1}{\sigma^2} J_{\mathrm{BS}}(y K/V_0) \right)\,,
\end{equation}
with
\begin{equation}
F_p(y) := \frac12 \left(I_1^c(g_p(y))\right)^2 
\left[ 1 - \left(\eta(g_p(y))\right)^2 y  \right]\,, 
\end{equation}
where $g_p(y)$ is the solution of the equation
\begin{equation*}
\frac{I_2^p(g_p)}{I_1^p(g_p)} = y,
\end{equation*}
where
\begin{align}
I_1^p(g_p) := \int_{S_0}^{g_p} \frac{\eta(x) dx}{x \sqrt{\eta^2(x) - \eta^2(g_p)} }\,,\quad
I_2^p(g_p) := \int_{S_0}^{g_p} \frac{dx}{x\eta(x) \sqrt{\eta^2(x) - \eta^2(g_p)} }\,.
\end{align}

This formulation follows from Proposition~\ref{prop:variational} by the change of variable $y = z/K$.

The local volatility function is chosen as the Tanh model \eqref{loc:vol:eta:S}
with model parameters
\begin{equation}\label{params1}
V_0 = 1.0, \quad\sigma=2.0,\quad f_0=1,\quad f_1=-0.5,\quad x_0=0\,.
\end{equation}
The choice of $V_0$ is chosen such that the numerical value of the parameter is sufficiently large $V_0 T \sim O(1)$ for maturities $T\sim 1$ year. 
This ensures that the testing is distinctive from that of pricing Asian options in the local volatility model.
If this parameter is small $V_0 T \ll 1$ then the model predictions are effectively the same as those for Asian options in the local volatility model with local volatility $\eta(x)$. This case was covered by \cite{PZAsian}.

The left plot in Figure~\ref{Fig:Rho0} compares the exact rate function $\mathcal{I}_0(S_0,V_0,K)$ (solid black) with the approximation of Proposition \ref{prop:rho} by expansion to $O(x^3)$ (dashed blue). 
The exact result for the rate function $\mathcal{I}_0(S_0,V_0,K)$ is used to compute the asymptotic implied volatility of the variance options $\Sigma_V(K)$ which is shown in the right plot of Figure \ref{Fig:Rho0}
as a function of log-moneyness $x=\log(K/(\eta_0^2 V_0))$ (solid black curve). This is compared with the linear approximation by an expansion in log-moneyness to $O(x^3)$ from Proposition~\ref{prop:rho}, see \eqref{SigVlin} (dotted black).

From this plot we conclude that the approximation by expansion in $x$ is reasonably good for sufficiently small log-moneyness, for options close to the ATM point. The approximation error increases away from this point.

\begin{figure}[htbp!]
\centering
\includegraphics[scale=0.55]{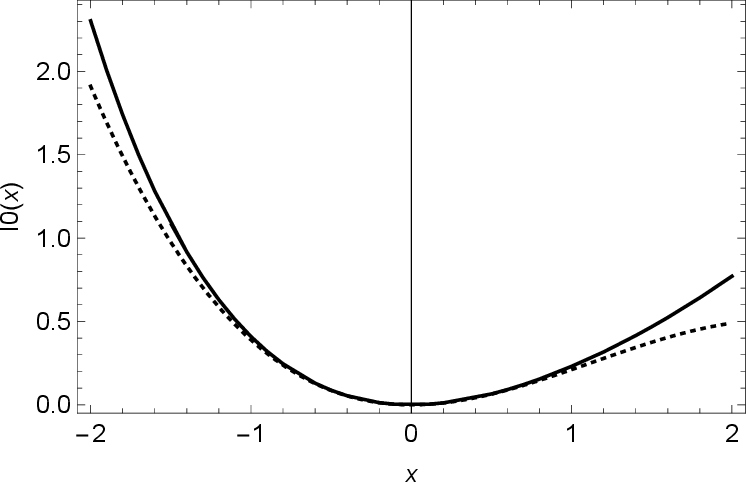}
\includegraphics[scale=0.55]{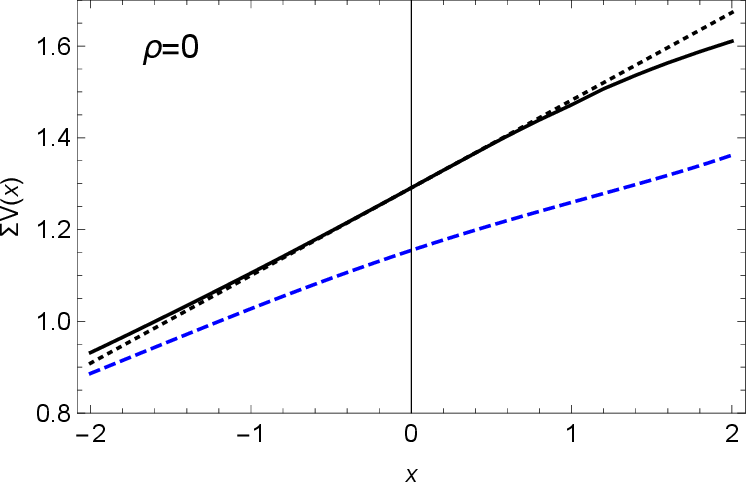}
\caption{Left: the exact rate function $\mathcal{I}_0(S_0,V_0,K)$ vs log-moneyness $x=\log(K/(\eta_0^2 V_0)$ from Proposition \ref{prop:variational} (solid black), comparing with the expansion to $O(x^3)$ from Proposition \ref{prop:rho} (dotted black).
Right: Asymptotic implied volatility $\Sigma_V(K)$ vs $x$, for the 
uncorrelated case $\rho=0$. 
The exact result from Proposition \ref{prop:variational} 
is shown as the solid black curve, and the linear approximation obtained by expanding in log-moneyness to $O(x)$ is the dotted black line. 
The dashed blue curve is the lower bound \eqref{bound2}. Model parameters are as in \eqref{params1}. }
\label{Fig:Rho0}
\end{figure}

\subsection{Bounds}

We discuss further in this section the two bounds on the rate function $\mathcal{I}_\rho(S_0,V_0,K)$ proved in Section~\ref{sec:main:results}.
Proposition~\ref{prop:bounds} gives upper and lower bounds for the correlated case in terms of the rate function for $\rho=0$. 
Proposition~\ref{prop:bounds:3} gives an upper bound in terms of the rate function $\mathcal{J}$ for Asian options in the local volatility model.
We present in this section several tests for the tightness of these bounds.

\textbf{Bound 1.}
Proposition~\ref{prop:bounds} gives lower and upper bounds on $I_\rho(K)$ in terms of $I_0(K)$. They can be translated into bounds on the implied volatility $\Sigma_V(K)$:
\begin{equation}\label{bound1}
\sqrt{1- |\rho| } \Sigma_V(K;\rho=0) \leq \Sigma_V(K) \leq \sqrt{1 + |\rho| } \Sigma_V(K;\rho=0)\,.
\end{equation}

Figure~\ref{Fig:b} shows these bounds for $\rho = \pm 0.5$ (left) and $\rho=\pm 0.7$ (right) as the solid blue and red curves, for the parameters given in \eqref{params1}. The bounds depend only on $|\rho|$ so they are identical for positive and negative correlation.
The dotted black lines show the linear approximations for $\Sigma_V(K)$ following from Proposition~\ref{prop:first:order}. The bounds are well satisfied for a range of strikes around the ATM point, although the upper bound is violated by the linear approximation at the upper range of the strikes considered.

\textbf{Bound 2.} Proposition~\ref{prop:bounds:3} gives an upper bound on $I_\rho(K)$ in terms of $\mathcal{J}(V_0,K/\eta_0^2)$, the rate function for Asian options in the local volatility model with volatility function $\sigma(v)$. This translates into a lower bound on the asymptotic implied volatility of variance options $\Sigma_V(K)$.

For constant vol-of-vol $\sigma(v) \equiv \sigma$, this bound reads
\begin{equation}\label{bound2}
\Sigma_V(K) \geq \sigma\sqrt{1-\rho^2} \Sigma_{\mathrm{LN}}^{(\mathrm{BS})}\left(K/(\eta_0^2 V_0)\right)\,,
\end{equation}
where $\Sigma_{\mathrm{LN}}^{(\mathrm{BS})}(k) = \frac{|\log k|}{\sqrt{2J_{\mathrm{BS}}(k)}}$ is the asymptotic
implied log-normal volatility for Asian options in the Black-Scholes model, expressed in terms of $J_{\mathrm{BS}}(k)$ introduced in \eqref{JBS:def}. 
For $k$ sufficiently close to 1 (near-ATM Asian options), $\Sigma_{\mathrm{LN}}^{(\mathrm{BS})}(k)$ is well approximated by the series expansion (see equation (55) in \cite{PZAsian})
\begin{equation}
\Sigma^{\mathrm{(BS)}}_{LN}(k) = \frac{1}{\sqrt3} \Big( 1 + \frac{1}{10} \log k - \frac{23}{2100} \log^2 k + O(\log^3 k)
\Big)\,.
\end{equation}

The dashed blue curve in 
Figure~\ref{Fig:Rho0} (right panel) shows the lower bound \eqref{bound2} for $\rho=0$, comparing with the exact asymptotic implied volatility $\Sigma_V(K)$ (solid black curve). 
The model parameters are as in \eqref{params1}.
The bound is indeed satisfied in the range of log-moneyness covered by the plot. 

The lower bound \eqref{bound2} is shown also in Figure~\ref{Fig:b} as the dashed blue curve, for several choices of the correlation $\rho = \pm 0.5$ (left) and $\rho = \pm 0.7$ (right). The bound is independent of the sign of $\rho$; thus we show on the same plot both cases with either sign for $\rho$. 
This bound is violated by the linear approximation for $\rho=-0.7$ for positive log-moneyness larger than about $0.5$. The bound is satisfied for all strikes by the linear approximation for $\rho = -0.5$. At the ATM point $K_{\mathrm{ATM}}=\eta_0^2 V_0$ one can see that the bound \eqref{bound2} is automatically satisfied by the asymptotic prediction for the ATM volatility \eqref{SigVATM}. Indeed, this bound reads
\begin{equation}
\Sigma_V(K_{\mathrm{ATM}}) = \frac{1}{\sqrt3}\sigma 
\sqrt{\sigma^2 + 4\eta_1 \rho \sigma \sqrt{V_0} + 4 \eta_1^2 V_0} \geq \frac{1}{\sqrt3} \sigma (1-\rho^2)\,,
\end{equation}
which is equivalent to $\rho^2 \sigma^2 + 4\eta_1 \rho \sigma \sqrt{V_0} + 4\eta_1^2 V_0 = (\rho \sigma + 2\eta_1 \sqrt{V_0})^2 \geq 0$, which holds indeed. 
By continuity, this implies that the lower bound \eqref{bound2} is always satisfied in a non-vanishing region around the ATM point.

\begin{figure}[htbp!]
\centering
\includegraphics[scale=0.5]{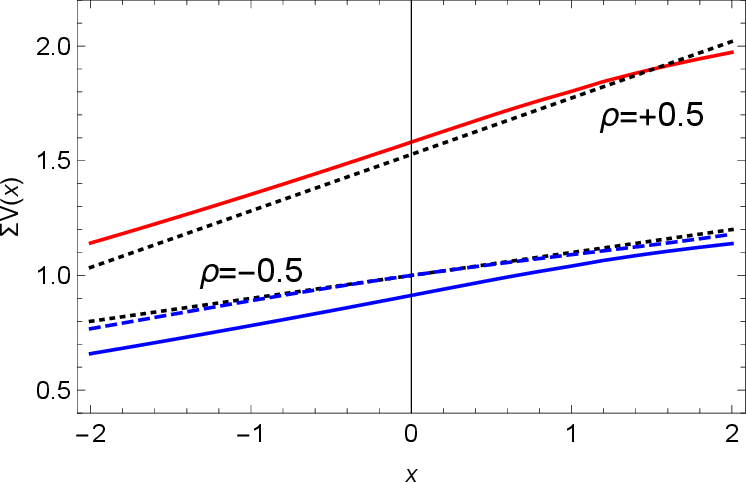}
\includegraphics[scale=0.5]{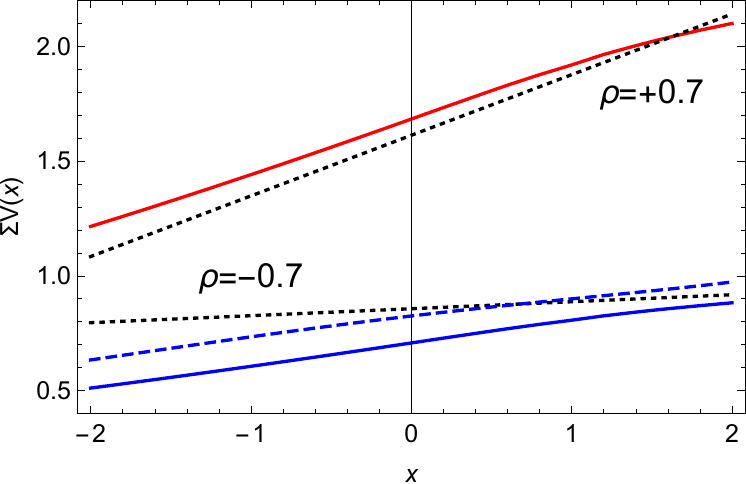}
\caption{Lower and upper bounds on the implied volatility $\Sigma_V(K)$ from Proposition \ref{prop:bounds} (bound 1) 
(solid blue and red curves), and lower bound from Proposition \ref{prop:bounds:3} (bound 2)
(dashed blue). The dotted black lines show the linear approximation for 
$\Sigma_V(K)$. Left: $\rho = \pm 0.5$, right: $\rho=\pm 0.7$.}
\label{Fig:b}
\end{figure}


\subsection{Monte Carlo tests}

The result of Proposition~\ref{prop:first:order} can be reformulated as a prediction for the ATM implied volatility $\Sigma_{V,\mathrm{ATM}}$
and ATM skew $s_V$ of the variance options.
They read explicitly
\begin{equation}\label{sigV}
\Sigma_{V,\mathrm{ATM}} = \frac{1}{\sqrt3} \sqrt{\sigma^2 + 4\eta_1 \sigma \rho \sqrt{V_0} + 4\eta_1^2 V_0}\,,
\end{equation}
and 
\begin{align}\label{sV}
s_V &= \frac{1}{10 \sqrt3 \left(\sigma^2 + 4\eta_1 \sigma \rho \sqrt{V_0} + 4\eta_1^2 V_0\right)^{3/2}} \\
&\qquad\qquad\cdot \Big( \sigma^4 + 14\sigma^3 \eta_1 \rho \sqrt{V_0} 
+ 4 \sigma^2 V_0\left( 6\eta_0\eta_2 \rho^2 + \eta_1^2 \left(5 + 7\rho^2\right)\right)  \nonumber \\
& \qquad\qquad\qquad\qquad + 8 \sigma \eta_1 \rho V_0^{3/2} \left(7 \eta_1^2 + 12 \eta_0 \eta_2\right) 
+ 16 \eta_1^2 V_0^2 \left(\eta_1^2 + 6\eta_0 \eta_2\right) \Big)\,.\nonumber
\end{align}
The result \eqref{sigV} coincides with the prediction of Theorem~\ref{thm:ATM}
for the short-maturity asymptotics of ATM variance options. 

Using these predictions, we construct the following linear approximation for the asymptotic
implied volatility of the variance options using the ATM level and skew given in \eqref{sigV} and \eqref{sV}:
\begin{equation}\label{SigVlin}
\Sigma_V^{\rm lin}(K) := \Sigma_{V,\mathrm{ATM}} + s_V x\,.
\end{equation}
This is expected to give a good approximation for $\Sigma_V(K)$ for strikes sufficiently close to the ATM point.
We will test this approximation for the implied volatility by comparing it with the numerical pricing of the variance options by Monte Carlo simulation.

We will use for this numerical test the Tanh local volatility model \eqref{loc:vol:eta:S}
with model parameters $f_0=1,f_1=-0.1, x_0=0$.
For the $V_t$ process, we assume $\sigma = 2.0, V_0 = 0.1$.
The spot asset price is $S_0=1$. The timeline is discretized by $N_T=2000$ steps, and we use $N_{\mathrm{MC}}=10^5$ MC paths for the simulation.

\begin{table}[h!]
  \centering
  \caption{The ATM volatility level, skew and convexity 
  for the short-maturity asymptotics of the variance options in the Tanh
  model. The last column shows the forward $F_V(T)$ for $T=1/12$ computed by MC simulation.}
    \begin{tabular}{|c|cc|c|}
    \hline
$\rho$ & $\sigma_{V,\mathrm{ATM}}$ & $s_V$  & $F_V(T)$ \\
               \hline\hline
$-0.7$ & 1.1806 & 0.1257 & $0.1004 \pm 0.0001$ \\
0     & 1.1553 & 0.1553 & $1.0006 \pm 0.0001$ \\
$+0.7$ & 1.1294 & 0.1053 & $0.0997 \pm 0.0001$ \\
    \hline
    \end{tabular}%
  \label{tab:params}%
\end{table}%

\begin{figure}[h]
\centering
\includegraphics[width=1.9in]{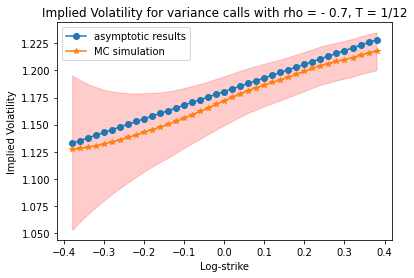}
\includegraphics[width=1.9in]{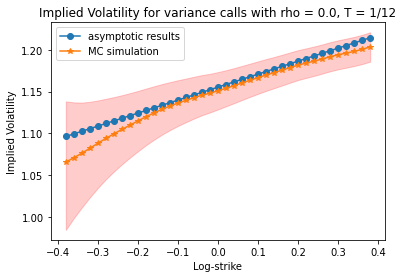}
\includegraphics[width=1.9in]{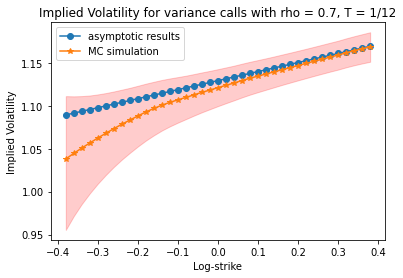}
\caption{The implied volatility of options on variance with maturity $T=1/12$ (1 month) in the Tanh LSV model, for three values of the correlation $\rho \in \{-0.7, 0, 0.7\}$. The orange points denote the MC simulation and the blue dots show the asymptotic prediction (\ref{SigVlin}).}
\label{Fig:Tanh}
\end{figure}

The MC simulation results for the implied volatility of variance options are shown in Figure~\ref{Fig:Tanh} for several choices of $\rho$ for variance options with maturity $T=1/12$ (orange dots). The bands show the MC errors, defined as $\frac{1}{\sqrt{N_{\mathrm{MC}}}} \mathrm{sd}(\mathrm{pay})$ in terms of the standard deviation of the option payoff for the MC sample.

The asymptotic prediction for the ATM implied vol from the linear approximation (\ref{SigVlin}) is shown as the blue dots in Figure~\ref{Fig:Tanh}. 
The numerical values of the ATM implied vol and the ATM skew from (\ref{sigV}) and (\ref{sV}) for each test case are shown in Table~\ref{tab:params}.
The asymptotic prediction agrees well with the MC simulation for a range of strikes around the ATM point, although it slightly overestimates the simulation result at the ATM point for $\rho = \pm 0.7$.

In order to investigate the size of the subleading corrections of $O(T)$, we show in Figure~\ref{Fig:Tanh1Day} the same results for variance options with maturity $T=1/252$ (1 business day), and the same model parameters as in Figure~\ref{Fig:Tanh}. These plots we compare the MC simulation with the asymptotic prediction in a more narrow range of log-moneyness $x\in [-0.1,+0.1]$. For this case the difference between the asymptotic prediction and the simulation results is much smaller. We conclude that the difference between the asymptotic result and the MC simulation in Figure~\ref{Fig:Tanh} can be attributed to subleading $O(T)$ corrections to the asymptotic result.

\begin{figure}[h]
\centering
\includegraphics[width=1.9in]{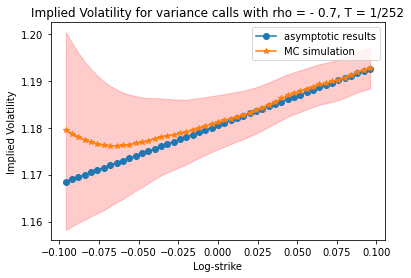}
\includegraphics[width=1.9in]{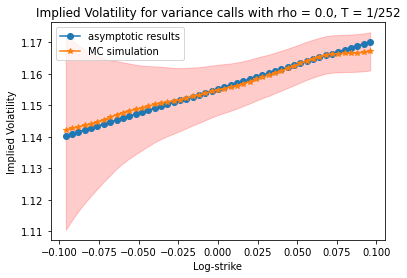}
\includegraphics[width=1.9in]{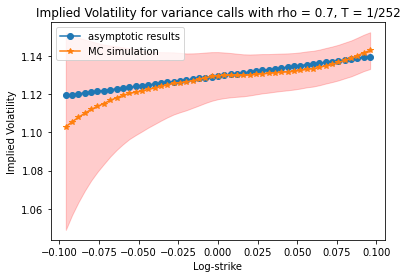}
\caption{Same as Figure \ref{Fig:Tanh} but the maturity of the variance options is  $T=1/252$ (1 day). }
\label{Fig:Tanh1Day}
\end{figure}

\section*{Acknowledgements}
The authors thank the Associate Editor 
and two anonymous referees for helpful comments and suggestions.
Xiaoyu Wang is supported by the Guangzhou-HKUST(GZ) Joint Funding Program (No.2024A03J0630), 
Guangzhou Municipal Key Laboratory of Financial Technology Cutting-Edge Research.
Lingjiong Zhu is partially supported by grants NSF DMS-2053454 and NSF DMS-2208303.

\bibliographystyle{plain}
\bibliography{ShortMaturityVariance}

\appendix


\section{Background on Large Deviations Theory}\label{sec:LDP}

We give in this Appendix a few basic concepts of large deviations theory from probability theory
which are in the proofs.
We refer to Dembo and Zeitouni \cite{Dembo1998} and Varadhan \cite{VaradhanLD} for more details on large deviations theory and its applications.

\begin{definition}[Large Deviation Principle]
A sequence $(P_\epsilon)_{\epsilon \in \mathbb{R}^+}$ of probability measures
on a topological space $X$ satisfies the large deviation principle with rate function $I: X \to \mathbb{R}$
if $I$ is non-negative, lower semicontinuous and for any measurable set $A$, we have
\begin{equation}
- \inf_{x\in A^o} I(x) \leq \liminf_{\epsilon\to 0} \epsilon \log P_\epsilon(A) \leq
\limsup_{\epsilon\to 0} \epsilon \log P_\epsilon(A) \leq - \inf_{x\in \bar A} I(x) \,,
\end{equation}
where $A^o$ denotes the interior of $A$ and $\bar A$ its closure.
\end{definition}

\begin{theorem}[Contraction Principle, see e.g. Theorem 4.2.1. in \cite{Dembo1998}]\label{Contraction:Thm}
If $F:X\rightarrow Y$ is a continuous map and 
$P_{\epsilon}$ satisfies a large deviation principle on $X$ with the rate 
function $I(x)$,
then the probability measures $Q_{\epsilon}:=P_{\epsilon}F^{-1}$ satisfies
a large deviation principle on $Y$ with the rate function
$J(y)=\inf_{x: F(x)=y}I(x)$.
\end{theorem}

\section{Technical Proofs}\label{sec:proofs}

\begin{proof}[Proof of Theorem~\ref{thm:OTM}]
Let us consider OTM case for call options, that is, $K>V_{0}\eta^{2}(S_{0})$.
The case for the put options is similar and the proof is omitted here.
First, it is easy to see that
\begin{equation}
\lim_{T\rightarrow 0}T\log C(T)
=\lim_{T\rightarrow 0}T\log\mathbb{E}\left[\left(\frac{1}{T}\int_{0}^{T}V_{s}\eta^{2}(S_{s})ds-K\right)^{+}\right],
\end{equation}
if the limit exists.
Next, we will show that
\begin{equation}\label{eqn:equality}
\lim_{T\rightarrow 0}T\log\mathbb{E}\left[\left(\frac{1}{T}\int_{0}^{T}V_{s}\eta^{2}(S_{s})ds-K\right)^{+}\right]
=\lim_{T\rightarrow 0}T\log\mathbb{Q}\left(\frac{1}{T}\int_{0}^{T}V_{s}\eta^{2}(S_{s})ds\geq K\right),
\end{equation}
if the limit exists.
The equality in \eqref{eqn:equality} can be established by considering the upper bound, i.e. $\limsup$ 
on the left hand side of \eqref{eqn:equality} and the lower bound, i.e. $\liminf$ on the left hand side
of \eqref{eqn:equality}. The argument for the lower bound is standard and is omitted here
The argument for the upper bound can be established via the following estimates.
For any $p,q>1$ with $\frac{1}{p}+\frac{1}{q}=1$, by H\"{o}lder's inequality, we have
\begin{align}
&\mathbb{E}\left[\left(\frac{1}{T}\int_{0}^{T}V_{s}\eta^{2}(S_{s})ds-K\right)^{+}\right]
\nonumber
\\
&=\mathbb{E}\left[\left(\frac{1}{T}\int_{0}^{T}V_{s}\eta^{2}(S_{s})ds-K\right)1_{\frac{1}{T}\int_{0}^{T}V_{s}\eta^{2}(S_{s})ds\geq K}\right]
\nonumber
\\
&\leq\left(\mathbb{E}\left[\left|\frac{1}{T}\int_{0}^{T}V_{s}\eta^{2}(S_{s})ds-K\right|^{p}\right]\right)^{\frac{1}{p}}\left(\mathbb{E}\left[\left(1_{\frac{1}{T}\int_{0}^{T}V_{s}\eta^{2}(S_{s})ds\geq K}\right)^{q}\right]\right)^{\frac{1}{q}}
\nonumber
\\
&=\left(\mathbb{E}\left[\left|\frac{1}{T}\int_{0}^{T}V_{s}\eta^{2}(S_{s})ds-K\right|^{p}\right]\right)^{\frac{1}{p}}\left(\mathbb{Q}\left(\frac{1}{T}\int_{0}^{T}V_{s}\eta^{2}(S_{s})ds\geq K\right)\right)^{\frac{1}{q}}.\label{ineq:holder}
\end{align}
For any $p>1$, we can show that
\begin{align}
\mathbb{E}\left[\left|\frac{1}{T}\int_{0}^{T}V_{s}\eta^{2}(S_{s})ds-K\right|^{p}\right]
&\leq
\mathbb{E}\left[\left(\frac{1}{T}\int_{0}^{T}V_{s}\eta^{2}(S_{s})ds+K\right)^{p}\right]
\nonumber
\\
&\leq
2^{p-1}\mathbb{E}\left[\left(\frac{1}{T}\int_{0}^{T}V_{s}\eta^{2}(S_{s})ds\right)^{p}+K^{p}\right]
\nonumber
\\
&\leq
2^{p-1}\left(\frac{1}{T}\int_{0}^{T}\mathbb{E}\left[V_{s}^{p}\eta^{2p}(S_{s})\right]ds+K^{p}\right),
\end{align}
where we used Jensen's inequality.

Furthermore, we can compute that
\begin{align}
\frac{1}{T}\int_{0}^{T}\mathbb{E}\left[V_{s}^{p}\eta^{2p}(S_{s})\right]ds
\leq
\frac{M_{\eta}^{2p}}{T}\int_{0}^{T}\mathbb{E}\left[V_{s}^{p}\right]ds,
\end{align}
and for any $0\leq s\leq T$,
\begin{align}
\mathbb{E}\left[V_{s}^{p}\right]
&=V_{0}^{p}\mathbb{E}\left[e^{\int_{0}^{s}\left(p\mu(V_{u})-\frac{p}{2}\sigma^{2}(V_{u})\right)du+\int_{0}^{s}p\sigma(V_{u})dZ_{u}}\right]
\nonumber
\\
&\leq
V_{0}^{p}e^{pM_{\mu}s+\frac{p^{2}}{2}M_{\sigma}^{2}s}\mathbb{E}\left[e^{\int_{0}^{s}\frac{-p^{2}}{2}\sigma^{2}(V_{u})du+\int_{0}^{s}p\sigma(V_{u})dZ_{u}}\right]
\nonumber
\\
&=V_{0}^{p}e^{pM_{\mu}s+\frac{p^{2}}{2}M_{\sigma}^{2}s}.\label{V:p:bound}
\end{align}
Therefore, we conclude that for any $p>1$,
\begin{align}
\mathbb{E}\left[\left|\frac{1}{T}\int_{0}^{T}V_{s}\eta^{2}(S_{s})ds-K\right|^{p}\right]
\leq
2^{p-1}\left(M_{\eta}^{2p}V_{0}^{p}e^{pM_{\mu}T+\frac{p^{2}}{2}M_{\sigma}^{2}T}+K^{p}\right).\label{p:finiteness}
\end{align}
Hence, by \eqref{ineq:holder} and \eqref{p:finiteness}, we showed that
\begin{equation}\label{q:eqn}
\limsup_{T\rightarrow 0}T\log\mathbb{E}\left[\left(\frac{1}{T}\int_{0}^{T}V_{s}\eta^{2}(S_{s})ds-K\right)^{+}\right]
\leq\limsup_{T\rightarrow 0}\frac{T}{q}\log\mathbb{Q}\left(\frac{1}{T}\int_{0}^{T}V_{s}\eta^{2}(S_{s})ds\geq K\right).
\end{equation}
Since \eqref{q:eqn} holds
for any $q>1$, we conclude that
\begin{equation}
\limsup_{T\rightarrow 0}T\log\mathbb{E}\left[\left(\frac{1}{T}\int_{0}^{T}V_{s}\eta^{2}(S_{s})ds-K\right)^{+}\right]
\leq\limsup_{T\rightarrow 0}T\log\mathbb{Q}\left(\frac{1}{T}\int_{0}^{T}V_{s}\eta^{2}(S_{s})ds\geq K\right).
\end{equation}
Therefore, we established the upper bound for \eqref{eqn:equality}.
Next, let us show that the limit on the right hand side of \eqref{eqn:equality} exists, 
which can be established through large deviations theory.

Under Assumptions~\ref{assump:bounded} and \ref{assump:LDP}, by the sample-path large deviations for 
small time diffusions (see for example \cite{Varadhan} and \cite{Robertson2010}),
one can see that $\mathbb{Q}(\{(\log S_{tT},\log V_{tT}),0\leq t\leq 1\}\in\cdot)$
satisfies a sample-path large deviation principle with the rate function:
\begin{equation}\label{rate:function:LDP}
\frac{1}{2(1-\rho^{2})}\int_{0}^{1}\left(\frac{g'(t)}{\eta(e^{g(t)})\sqrt{e^{h(t)}}}-\frac{\rho h'(t)}{\sigma(e^{h(t)})}\right)^{2}dt
+\frac{1}{2}\int_{0}^{1}\left(\frac{h'(t)}{\sigma(e^{h(t)})}\right)^{2}dt,    
\end{equation}
with $g(0)=\log S_{0}$, $h(0)=\log V_{0}$ and $g,h$ being absolutely continuous
and the rate function is $+\infty$ otherwise.

By an application of the contraction principle (see for example Theorem~4.2.1. in \cite{Dembo1998}, restated in Theorem~\ref{Contraction:Thm}),
\begin{align}
&\lim_{T\rightarrow 0}T\log\mathbb{Q}\left(\frac{1}{T}\int_{0}^{T}V_{s}\eta^{2}(S_{s})ds\geq K\right)
\nonumber
\\
&=\lim_{T\rightarrow 0}T\log\mathbb{Q}\left(\int_{0}^{1}V_{tT}\eta^{2}(S_{tT})dt\geq K\right)
\nonumber
\\
&=-\inf_{\substack{g(0)=\log S_{0},h(0)=\log V_{0}\\
\int_{0}^{1}e^{h(t)}\eta^{2}(e^{g(t)})dt=K}}
\Bigg\{\frac{1}{2(1-\rho^{2})}\int_{0}^{1}\left(\frac{g'(t)}{\eta(e^{g(t)})\sqrt{e^{h(t)}}}-\frac{\rho h'(t)}{\sigma(e^{h(t)})}\right)^{2}dt
\nonumber
\\
&\qquad\qquad\qquad\qquad\qquad\qquad\qquad\qquad
+\frac{1}{2}\int_{0}^{1}\left(\frac{h'(t)}{\sigma(e^{h(t)})}\right)^{2}dt\Bigg\}.
\end{align}
Hence, we conclude that
\begin{align}
&\lim_{T\rightarrow 0}T\log\mathbb{E}\left[\left(\frac{1}{T}\int_{0}^{T}V_{s}\eta^{2}(S_{s})ds-K\right)^{+}\right]
\nonumber
\\
&=-\inf_{\substack{g(0)=\log S_{0},h(0)=\log V_{0}\\
\int_{0}^{1}e^{h(t)}\eta^{2}(e^{g(t)})dt=K}}
\Bigg\{\frac{1}{2(1-\rho^{2})}\int_{0}^{1}\left(\frac{g'(t)}{\eta(e^{g(t)})\sqrt{e^{h(t)}}}-\frac{\rho h'(t)}{\sigma(e^{h(t)})}\right)^{2}dt
\nonumber
\\
&\qquad\qquad\qquad\qquad\qquad\qquad\qquad\qquad
+\frac{1}{2}\int_{0}^{1}\left(\frac{h'(t)}{\sigma(e^{h(t)})}\right)^{2}dt\Bigg\}.
\end{align}
This completes the proof.
\end{proof}


\begin{proof}[Proof of Proposition~\ref{prop:variational}]
(i) Let us first consider the OTM variance call options, i.e. $V_{0}\eta^{2}(S_{0})<K$.
When $\rho=0$, one can compute that
\begin{align}
&\mathcal{I}(S_{0},V_{0},K)
\nonumber
\\
&=\inf_{\substack{g(0)=\log S_{0},h(0)=\log V_{0}\\
\int_{0}^{1}e^{h(t)}\eta^{2}(e^{g(t)})dt=K}}
\left\{\frac{1}{2}\int_{0}^{1}\left(\frac{g'(t)}{\eta(e^{g(t)})\sqrt{e^{h(t)}}}\right)^{2}dt
+\frac{1}{2}\int_{0}^{1}\left(\frac{h'(t)}{\sigma(e^{h(t)})}\right)^{2}dt\right\}.
\end{align}
For this proof, we will proceed by assuming 
that the solutions $g,h$ of the Euler-Lagrange equation for this variational problem exist.
Some motivation for this assumption is offered by Remark~\ref{cor:variational} where we give explicit solutions for these functions.

Given $h$, let us define:
\begin{equation}
\Lambda[g]:=\frac{1}{2}\int_{0}^{1}\left(\frac{g'(t)}{\eta(e^{g(t)})\sqrt{e^{h(t)}}}\right)^{2}dt
+\lambda\left(\int_{0}^{1}e^{h(t)}\eta^{2}(e^{g(t)})dt-K\right),
\end{equation}
where $\lambda$ is the Lagrange multiplier. 
The Euler-Lagrange equation gives:
\begin{align}
-\frac{\eta'(e^{g})e^{g}(g')^{2}}{\eta^{3}(e^{g})e^{h}}
+\lambda e^{h}2\eta(e^{g})\eta'(e^{g})e^{g}
&=\frac{d}{dt}\left(\frac{g'}{\eta^{2}(e^{g})e^{h}}\right)
\nonumber
\\
&=\frac{g''}{\eta^{2}(e^{g})e^{h}}
-\frac{g'(2\eta(e^{g})\eta'(e^{g})e^{g}g'e^{h}+\eta^{2}(e^{g})e^{h}h')}{\eta^{4}(e^{g})e^{2h}},\label{eqn:g}
\end{align}
and equation~\eqref{eqn:g} implies that
\begin{align}
\frac{\eta'(e^{g})e^{g}(g')^{2}}{\eta^{3}(e^{g})e^{h}}
+\lambda e^{h}2\eta(e^{g})\eta'(e^{g})e^{g}
=\frac{g''}{\eta^{2}(e^{g})e^{h}}
-\frac{g'h'}{\eta^{2}(e^{g})e^{h}},
\end{align}
which is equivalent to
\begin{align}\label{E-L-g}
2\lambda e^{h}\eta^{2}(e^{g})\eta'(e^{g})e^{g}
&=\frac{g''}{\eta(e^{g})e^{h}}
-\frac{\eta'(e^{g})e^{g}(g')^{2}}{\eta^{2}(e^{g})e^{h}}-\frac{g'h'}{\eta(e^{g})e^{h}}
\nonumber
\\
&=\frac{d}{dt}\left(\frac{g'}{\eta(e^{g})e^{h}}\right).
\end{align}
This implies that
\begin{align}
\frac{1}{2}\frac{d}{dt}\left(\frac{g'}{\eta(e^{g})e^{h}}\right)^{2}
=\frac{g'}{\eta(e^{g})e^{h}}\frac{d}{dt}\left(\frac{g'}{\eta(e^{g})e^{h}}\right)
=2\lambda\eta(e^{g})\eta'(e^{g})e^{g}g'.
\end{align}
Integrating the above equation from $t=1$ to $t$ and taking into account the transversality condition $g'(1)=0$ gives
\begin{align}\label{gprime-sq}
\frac{1}{2}\left(\frac{g'(t)}{\eta(e^{g(t)})e^{h(t)}}\right)^{2}
=\lambda\left(\eta^{2}(e^{g(t)})-\eta^{2}(e^{g(1)})\right).
\end{align}
Therefore,
\begin{align}
\frac{1}{2}\left(\frac{g'(t)}{\eta(e^{g(t)})\sqrt{e^{h(t)}}}\right)^{2}
=\lambda\left(\eta^{2}(e^{g(t)})e^{h(t)}-\eta^{2}(e^{g(1)})e^{h(t)}\right).
\end{align}
By integrating the above equation from $t=0$ to $t=1$, we get
\begin{align}
\int_{0}^{1}\frac{1}{2}\left(\frac{g'(t)}{\eta(e^{g(t)})\sqrt{e^{h(t)}}}\right)^{2}dt
&=\lambda\left(\int_{0}^{1}\eta^{2}(e^{g(t)})e^{h(t)}dt-\eta^{2}(e^{g(1)})\int_{0}^{1}e^{h(t)}dt\right)
\nonumber
\\
&=\lambda\left(K-\eta^{2}(e^{g(1)})\int_{0}^{1}e^{h(t)}dt\right).\label{eqn:plug}
\end{align}

For OTM variance call option, $\eta^{2}(e^{g(1)})\int_{0}^{1}e^{h(t)}dt\geq K$,
which implies that $\lambda\leq 0$. 
Furthermore, the assumed monotonicity of $\eta(\cdot)$ implies a monotonicity property for $g(t)$. Integrating \eqref{E-L-g} from $t$ to $1$ and using the transversality condition $g'(1)=0$ gives
\begin{align}\label{gprime}
\frac{g'(t)}{\eta(e^{g(t)}) e^{h(t)}} =
- 2\lambda \int_t^1 e^{g(s)+h(s)} \eta^2(e^{g(s)}) \eta'(e^{g(s)}) ds
\end{align}
By assumption $\eta'(\cdot)\leq 0$ which implies that the integral is non-positive, and thus $g'(t)\leq 0$. 

From \eqref{gprime-sq} we obtain
\begin{equation}\label{g:prime}
g'(t)= - \eta(e^{g(t)})e^{h(t)}
\sqrt{2\lambda\left(\eta^{2}(e^{g(t)})-\eta^{2}(e^{g(1)})\right)}.
\end{equation}
This implies that
\begin{equation}
\frac{dg}{\eta(e^{g})\sqrt{2\lambda\left(\eta^{2}(e^{g})-\eta^{2}(e^{g(1)})\right)}}
= - e^{h(t)}dt.
\end{equation}
Therefore,
\begin{equation}
\int_{S_{0}}^{e^{g(t)}}\frac{dx}{x\eta(x)\sqrt{2\lambda\left(\eta^{2}(x)-\eta^{2}(e^{g(1)})\right)}}
= - \int_{0}^{t}e^{h(s)}ds.
\end{equation}
By letting $t=1$, we can see that $g(1)$ solves the equation
\begin{equation}\label{divide:1}
\int_{S_{0}}^{e^{g(1)}}\frac{dx}{x\eta(x)\sqrt{2\lambda\left(\eta^{2}(x)-\eta^{2}(e^{g(1)})\right)}}
= - \int_{0}^{1}e^{h(s)}ds.
\end{equation}

Moreover, from \eqref{g:prime}, we have
\begin{equation}
\frac{\eta(e^{g(t)})g'(t)}{\sqrt{2\lambda\left(\eta^{2}(e^{g(t)})-\eta^{2}(e^{g(1)})\right)}}
= - \eta^{2}(e^{g(t)})e^{h(t)}.
\end{equation}
By integrating from $t=0$ to $t=1$, we obtain
\begin{equation}
\int_{e^{g(0)}}^{e^{g(1)}}\frac{\eta(x)dx}{x\sqrt{2\lambda\left(\eta^{2}(x)-\eta^{2}(e^{g(1)})\right)}}
= - \int_{0}^{1}\eta^{2}(e^{g(t)})e^{h(t)}dt,
\end{equation}
which implies that
\begin{equation}\label{divide:2}
\int_{S_{0}}^{e^{g(1)}}\frac{\eta(x)dx}{x\sqrt{2\lambda\left(\eta^{2}(x)-\eta^{2}(e^{g(1)})\right)}}
= - K.
\end{equation}
This yields that
\begin{equation}\label{lambda:eqn}
\sqrt{-\lambda}
=\frac{1}{- K}\int_{S_{0}}^{e^{g(1)}}\frac{\eta(x)dx}{x\sqrt{2\left(\eta^{2}(e^{g(1)})-\eta^{2}(x)\right)}}.
\end{equation}

By plugging \eqref{lambda:eqn} into \eqref{eqn:plug}, we get
\begin{align}
&\int_{0}^{1}\frac{1}{2}\left(\frac{g'(t)}{\eta(e^{g(t)})\sqrt{e^{h(t)}}}\right)^{2}dt
\nonumber
\\
&=\frac{1}{K^{2}}
\left(\int_{S_{0}}^{e^{g(1)}}\frac{\eta(x)dx}{x\sqrt{2\left(\eta^{2}(e^{g(1)})-\eta^{2}(x)\right)}}\right)^{2}
\left(\eta^{2}(e^{g(1)})\int_{0}^{1}e^{h(t)}dt-K\right).
\end{align}

Finally, by dividing \eqref{divide:1} by \eqref{divide:2}, 
we get
\begin{equation}
\frac{\int_{S_{0}}^{e^{g(1)}}\frac{dx}{x\eta(x)\sqrt{\eta^{2}(e^{g(1)})-\eta^{2}(x)}}}
{\int_{S_{0}}^{e^{g(1)}}\frac{\eta(x)dx}{x\sqrt{\eta^{2}(e^{g(1)})-\eta^{2}(x)}}}
=\frac{\int_{0}^{1}e^{h(s)}ds}{K},
\end{equation}
which determines the value of $g(1)$.

Hence, we conclude that 
with fixed $h$, 
\begin{align}
&\inf_{g(0)=\log S_{0},
\int_{0}^{1}e^{h(t)}\eta^{2}(e^{g(t)})dt=K}
\frac{1}{2}\int_{0}^{1}\left(\frac{g'(t)}{\eta(e^{g(t)})\sqrt{e^{h(t)}}}\right)^{2}dt
\nonumber
\\
&=\frac{1}{K^{2}}
\left(\int_{S_{0}}^{e^{g(1)}}\frac{\eta(x)dx}{x\sqrt{2\left(\eta^{2}(e^{g(1)})-\eta^{2}(x)\right)}}\right)^{2}
\left(\eta^{2}(e^{g(1)})\int_{0}^{1}e^{h(t)}dt-K\right),
\end{align}
where $g(1)$ is a function of $z:=\int_{0}^{1}e^{h(s)}ds$
such that
\begin{align}
&\inf_{g(0)=\log S_{0},
\int_{0}^{1}e^{h(t)}\eta^{2}(e^{g(t)})dt=K}
\frac{1}{2}\int_{0}^{1}\left(\frac{g'(t)}{\eta(e^{g(t)})\sqrt{e^{h(t)}}}\right)^{2}dt
\nonumber
\\
&=\frac{1}{K^{2}}
\left(\int_{S_{0}}^{G_{c}(z)}\frac{\eta(x)dx}{x\sqrt{2\left(\eta^{2}(G_{c}(z))-\eta^{2}(x)\right)}}\right)^{2}
\left(\eta^{2}(G_{c}(z))z-K\right),
\end{align}
where $G_{c}(z):=e^{g(1)}$ is the solution of the equation
\begin{equation}
\frac{\int_{S_{0}}^{G_{c}(z)}\frac{dx}{x\eta(x)\sqrt{\eta^{2}(G_{c}(z))-\eta^{2}(x)}}}
{\int_{S_{0}}^{G_{c}(z)}\frac{\eta(x)dx}{x\sqrt{\eta^{2}(G_{c}(z))-\eta^{2}(x)}}}
=\frac{z}{K}.
\end{equation}
Hence, we conclude that
\begin{align}
\mathcal{I}(S_{0},V_{0},K)
&=\inf_{z}
\Bigg\{\frac{1}{K^{2}}
\left(\int_{S_{0}}^{G_{c}(z)}\frac{\eta(x)dx}{x\sqrt{2\left(\eta^{2}(G_{c}(z))-\eta^{2}(x)\right)}}\right)^{2}
\left(\eta^{2}(G_{c}(z))z-K\right)
\nonumber
\\
&\qquad\qquad\qquad\qquad\qquad
+\inf_{h(0)=\log V_{0},
\int_{0}^{1}e^{h(t)}dt=z}\frac{1}{2}\int_{0}^{1}\left(\frac{h'(t)}{\sigma(e^{h(t)})}\right)^{2}dt\Bigg\}.
\end{align}

(ii) Next, let us consider solving the variational problem for the OTM put options
(i.e. $K<V_{0}\eta^{2}(S_{0})$)
when the correlation $\rho=0$. 
In this case, one can compute that
\begin{align}
&\mathcal{I}(S_{0},V_{0},K)
\nonumber
\\
&=\inf_{\substack{g(0)=\log S_{0},h(0)=\log V_{0}\\
\int_{0}^{1}e^{h(t)}\eta^{2}(e^{g(t)})dt=K}}
\left\{\frac{1}{2}\int_{0}^{1}\left(\frac{g'(t)}{\eta(e^{g(t)})\sqrt{e^{h(t)}}}\right)^{2}dt
+\frac{1}{2}\int_{0}^{1}\left(\frac{h'(t)}{\sigma(e^{h(t)})}\right)^{2}dt\right\}.
\end{align}

Similar as in the case for call options, given $h$, we can
write down the Euler-Lagrange equation for $g$.
For OTM variance put option, $\eta^{2}(e^{g(1)})\int_{0}^{1}e^{h(t)}dt\leq K$ 
which implies $\lambda\geq 0$.
The assumed monotonicity property $\eta'(\cdot) \leq 0$ implies by equation \eqref{gprime} that for this case we have $g'(t)\geq 0$.
Thus we have 
\begin{equation}\label{g:prime:put}
g'(t)= \eta(e^{g(t)})e^{h(t)}
\sqrt{2\lambda\left(\eta^{2}(e^{g(t)})-\eta^{2}(e^{g(1)})\right)}.
\end{equation}
This implies that
\begin{equation}
\frac{dg}{\eta(e^{g})\sqrt{2\lambda\left(\eta^{2}(e^{g})-\eta^{2}(e^{g(1)})\right)}}
= e^{h(t)}dt.
\end{equation}
Therefore,
\begin{equation}
\int_{S_{0}}^{e^{g(t)}}\frac{dx}{x\eta(x)\sqrt{2\lambda\left(\eta^{2}(x)-\eta^{2}(e^{g(1)})\right)}}
=  \int_{0}^{t}e^{h(s)}ds.
\end{equation}
By letting $t=1$, we can see that $g(1)$ solves the equation
\begin{equation}\label{divide:1:put}
\int_{S_{0}}^{e^{g(1)}}\frac{dx}{x\eta(x)\sqrt{2\lambda\left(\eta^{2}(x)-\eta^{2}(e^{g(1)})\right)}}
= \int_{0}^{1}e^{h(s)}ds.
\end{equation}

Moreover, from \eqref{g:prime:put}, we have
\begin{equation}
\frac{\eta(e^{g(t)})g'(t)}{\sqrt{2\lambda\left(\eta^{2}(e^{g(t)})-\eta^{2}(e^{g(1)})\right)}}
= \eta^{2}(e^{g(t)})e^{h(t)}.
\end{equation}
By integrating from $t=0$ to $t=1$, we obtain
\begin{equation}
\int_{e^{g(0)}}^{e^{g(1)}}\frac{\eta(x)dx}{x\sqrt{2\lambda\left(\eta^{2}(x)-\eta^{2}(e^{g(1)})\right)}}
= \int_{0}^{1}\eta^{2}(e^{g(t)})e^{h(t)}dt,
\end{equation}
which implies that
\begin{equation}\label{divide:2:put}
\int_{S_{0}}^{e^{g(1)}}\frac{\eta(x)dx}{x\sqrt{2\lambda\left(\eta^{2}(x)-\eta^{2}(e^{g(1)})\right)}}
= K.
\end{equation}
This yields that
\begin{equation}\label{lambda:eqn:put}
\sqrt{\lambda}
=\frac{1}{ K}\int_{S_{0}}^{e^{g(1)}}\frac{\eta(x)dx}{x\sqrt{2\left(\eta^{2}(x)-\eta^{2}(e^{g(1)})\right)}}.
\end{equation}

By plugging \eqref{lambda:eqn:put} into \eqref{eqn:plug}, we get
\begin{align}
&\int_{0}^{1}\frac{1}{2}\left(\frac{g'(t)}{\eta(e^{g(t)})\sqrt{e^{h(t)}}}\right)^{2}dt
\nonumber
\\
&=\frac{1}{K^{2}}
\left(\int_{S_{0}}^{e^{g(1)}}\frac{\eta(x)dx}{x\sqrt{2\left(\eta^{2}(x)-\eta^{2}(e^{g(1)})\right)}}\right)^{2}
\left(K-\eta^{2}(e^{g(1)})\int_{0}^{1}e^{h(t)}dt\right).
\end{align}

Finally, by dividing \eqref{divide:1:put} by \eqref{divide:2:put}, 
we get
\begin{equation}
\frac{\int_{S_{0}}^{e^{g(1)}}\frac{dx}{x\eta(x)\sqrt{\eta^{2}(x)-\eta^{2}(e^{g(1)})}}}
{\int_{S_{0}}^{e^{g(1)}}\frac{\eta(x)dx}{x\sqrt{\eta^{2}(x)-\eta^{2}(e^{g(1)})}}}
=\frac{\int_{0}^{1}e^{h(s)}ds}{K},
\end{equation}
which determines the value of $g(1)$.

Hence, we conclude that 
with fixed $h$, 
\begin{align}
&\inf_{g(0)=\log S_{0},
\int_{0}^{1}e^{h(t)}\eta^{2}(e^{g(t)})dt=K}
\frac{1}{2}\int_{0}^{1}\left(\frac{g'(t)}{\eta(e^{g(t)})\sqrt{e^{h(t)}}}\right)^{2}dt
\nonumber
\\
&=\frac{1}{K^{2}}
\left(\int_{S_{0}}^{e^{g(1)}}\frac{\eta(x)dx}{x\sqrt{2\left(\eta^{2}(x)-\eta^{2}(e^{g(1)})\right)}}\right)^{2}
\left(K-\eta^{2}(e^{g(1)})\int_{0}^{1}e^{h(t)}dt\right),
\end{align}
where $g(1)$ is a function of $z:=\int_{0}^{1}e^{h(s)}ds$
such that
\begin{align}
&\inf_{g(0)=\log S_{0},
\int_{0}^{1}e^{h(t)}\eta^{2}(e^{g(t)})dt=K}
\frac{1}{2}\int_{0}^{1}\left(\frac{g'(t)}{\eta(e^{g(t)})\sqrt{e^{h(t)}}}\right)^{2}dt
\nonumber
\\
&=\frac{1}{K^{2}}
\left(\int_{S_{0}}^{G_{p}(z)}\frac{\eta(x)dx}{x\sqrt{2\left(\eta^{2}(x)-\eta^{2}(G_{p}(z))\right)}}\right)^{2}
\left(K-\eta^{2}(G_{p}(z))z\right),
\end{align}
where $G_{p}(z):=e^{g(1)}$ is the solution of the equation
\begin{equation}
\frac{\int_{S_{0}}^{G_{p}(z)}\frac{dx}{x\eta(x)\sqrt{\eta^{2}(x)-\eta^{2}(G_{p}(z))}}}
{\int_{S_{0}}^{G_{p}(z)}\frac{\eta(x)dx}{x\sqrt{\eta^{2}(x)-\eta^{2}(G_{p}(z))}}}
=\frac{z}{K}.
\end{equation}
Hence, we conclude that
\begin{align}
\mathcal{I}(S_{0},V_{0},K)
&=\inf_{z}
\Bigg\{\frac{1}{K^{2}}
\left(\int_{S_{0}}^{G_{p}(z)}\frac{\eta(x)dx}{x\sqrt{2\left(\eta^{2}(x)-\eta^{2}(G_{p}(z))\right)}}\right)^{2}
\left(K-\eta^{2}(G_{p}(z))z\right)
\nonumber
\\
&\qquad\qquad\qquad\qquad\qquad
+\inf_{h(0)=\log V_{0},
\int_{0}^{1}e^{h(t)}dt=z}\frac{1}{2}\int_{0}^{1}\left(\frac{h'(t)}{\sigma(e^{h(t)})}\right)^{2}dt\Bigg\}.
\end{align}
This completes the proof.
\end{proof}


We give here a few details of computation for the relation appearing in Remark \ref{cor:variational}.

\begin{proof}
(i) For OTM variance call options, i.e. $V_{0}\eta^{2}(S_{0})<K$,
let $z_{c}$ be the optimizer in \eqref{I:formula:call}. Then, 
from the proof of Proposition~\ref{prop:variational}, 
$g_{0}(1)=\log G_{c}(z_{c})$ and the optimal Lagrange multiplier is given by:
\begin{equation}
\lambda_{0}=-\frac{1}{K^{2}}\left(\int_{S_{0}}^{G_{c}(z_{c})}\frac{\eta(x)dx}{x\sqrt{2(\eta^{2}(G_{c}(z_{c}))-\eta^{2}(x))}}\right)^{2}.
\end{equation}
We also obtain from the proof of Proposition~\ref{prop:variational} that
\begin{equation}
\int_{S_{0}}^{e^{g_{0}(t)}}\frac{dx}{x\eta(x)\sqrt{2\lambda_{0}\left(\eta^{2}(x)-\eta^{2}(e^{g_{0}(1)})\right)}}
= - \int_{0}^{t}e^{h_{0}(s)}ds.
\end{equation}
For this case $\lambda_{0}<0$. As shown in the proof of Proposition \ref{prop:variational}, the assumed monotonicity of $\eta(.)$ implies that $g_0(t)$ is decreasing.
Substituting here \eqref{lambda:eqn} for $\sqrt{-\lambda_0}$ 
this is written equivalently as
\begin{equation}
\int_{S_{0}}^{e^{g_{0}(t)}}\frac{dx}{x\eta(x)\sqrt{2\left(\eta^{2}(G_{c}(z_{c}))-\eta^{2}(x)\right)}}
=\frac{\int_{S_{0}}^{G_{c}(z_{c})}\frac{\eta(x)dx}{x\sqrt{2(\eta^{2}(G_{c}(z_{c}))-\eta^{2}(x))}}}{K}\int_{0}^{t}e^{h_{0}(s)}ds.
\end{equation}
Moreover, the optimal $h_{0}(t)$ solves the variational problem:
\begin{equation}
\inf_{h(0)=\log V_{0},
\int_{0}^{1}e^{h(t)}dt=z_{c}}\frac{1}{2}\int_{0}^{1}\left(\frac{h'(t)}{\sigma(e^{h(t)})}\right)^{2}dt,
\end{equation}
which has already been solved in \cite{PZAsian}. For convenience we quote the solution below. 
Writing the optimal $h_{0}(t)=\log V_{0}+f_{0}(t;z)$, we have
\begin{equation}\label{31}
f'_{0}(t;z) = 
\begin{cases}
\sqrt{-2\lambda_{-}} \sigma(V_{0}e^{f_0(t;z)}) \sqrt{e^{f_{0}(1;z)} - e^{f_{0}(t;z)}} \,, & z>V_0\,, \lambda_{-}<0  
\\
-\sqrt{2\lambda_{+}} \sigma(V_{0}e^{f_0(t;z)}) \sqrt{e^{f_{0}(t;z)} - e^{f_{0}(1;z)}} \,, & z<V_0\,, \lambda_{+}>0
\end{cases},
\end{equation}
where $\lambda_{-}=\frac{-1}{2}(F^{(-)}(\alpha_{-}))^{2}$
and $\lambda_{+}=\frac{1}{2}(F^{(+)}(\alpha_{+}))^{2}$.

The parameters $\alpha_\pm$ and $\lambda_\pm$ are defined as follows.

a) $z\geq V_0$. For this case $\alpha_{-}=f_{0}(1;z) \geq  0$ is given by the solution of the equation
\begin{equation}
e^{\alpha_{-}} - \frac{z}{V_0} = \frac{G^{(-)}(\alpha_{-})}{F^{(-)}(\alpha_{-})},
\end{equation}
with
\begin{align}
G^{(-)}(\alpha_{-})
=\int_{0}^{\alpha_{-}} \frac{\sqrt{e^{\alpha_{-}} - e^{y}}}{\sigma(V_0 e^{y})} dy,
\quad
F^{(-)}(\alpha_{-})=\int_{0}^{\alpha_{-}} \frac{1}{\sigma(V_0 e^{y})} 
\frac{1}{\sqrt{e^{\alpha_{-}} -  e^{y}}}dy \,.
\end{align}

b) $z \leq V_0$. For this case
$\alpha_{+}=-f_{0}(1;z) \geq 0$ is the solution of the equation
\begin{equation}
\frac{z}{V_0} - e^{-\alpha_{+}} = \frac{G^{(+)}(\alpha_{+})}{F^{(+)}(\alpha_{+})},
\end{equation}
with
\begin{align}
G^{(+)}(\alpha_{+}) = \int_{0}^{\alpha_{+}} \frac{ \sqrt{e^{-y} - e^{-\alpha_{+}}}}{\sigma( V_0 e^{-y})} dy, \quad
F^{(+)}(\alpha_{+}) = \int_{0}^{\alpha_{+}} \frac{1}{\sigma( V_0 e^{-y})} 
\frac{1}{\sqrt{e^{-y} - e^{-\alpha_{+}}}}dy \,,
\end{align}

Finally, we can solve \eqref{31} to obtain
that when $z>V_{0}$ 
\begin{equation}
\int_{0}^{f_{0}(t;z)}\frac{dy}{\sigma(V_{0}e^{y})\sqrt{e^{\alpha_{-}}-e^{y}}}=F^{(-)}(\alpha_{-})t,
\end{equation}
and when $z<V_{0}$, 
\begin{equation}
\int_{0}^{f_{0}(t;z)}\frac{dy}{\sigma(V_{0}e^{y})\sqrt{e^{y}-e^{-\alpha_{+}}}}=-F^{(+)}(\alpha_{+})t.
\end{equation}

(ii) For OTM variance put options, i.e. $V_{0}\eta^{2}(S_{0})>K$,
let $z_{p}$ be the optimizer in \eqref{I:formula}. Then, 
from the proof of Proposition~\ref{prop:variational}, 
$g_{0}(1)=\log G_{p}(z_{p})$ and the optimal Lagrange multiplier is given by:
\begin{equation}
\lambda_{0}=\frac{1}{K^{2}}\left(\int_{S_{0}}^{G_{p}(z_{p})}\frac{\eta(x)dx}{x\sqrt{2(\eta^{2}(x)-\eta^{2}(G_{p}(z_{p})))}}\right)^{2}.
\end{equation}
We also obtain from the proof of Proposition~\ref{prop:variational} that
\begin{equation}
\int_{S_{0}}^{e^{g_{0}(t)}}\frac{dx}{x\eta(x)\sqrt{2\lambda_{0}\left(\eta^{2}(x)-\eta^{2}(e^{g_{0}(1)})\right)}}
= \int_{0}^{t}e^{h_{0}(s)}ds.
\end{equation}
For this case $\lambda_{0}>0$ and $g_0(t)$ is increasing.

Using \eqref{lambda:eqn:put} for $\sqrt{\lambda_0}$ this is written equivalently 
as
\begin{equation}
\int_{S_{0}}^{e^{g_{0}(t)}}\frac{dx}{x\eta(x)\sqrt{2\left(\eta^{2}(x)-\eta^{2}(G_{p}(z_{p}))\right)}}
=\frac{\int_{S_{0}}^{G_{p}(z_{p})}\frac{\eta(x)dx}{x\sqrt{2(\eta^{2}(x)-\eta^{2}(G_{p}(z_{p})))}}}{K}\int_{0}^{t}e^{h_{0}(s)}ds.
\end{equation}
Moreover, the optimal $h_{0}(t)$ solves the variational problem:
\begin{equation}
\inf_{h(0)=\log V_{0},
\int_{0}^{1}e^{h(t)}dt=z_{p}}\frac{1}{2}\int_{0}^{1}\left(\frac{h'(t)}{\sigma(e^{h(t)})}\right)^{2}dt,
\end{equation}
and the optimal $h_{0}(t)=\log V_{0}+f_{0}(t;z_{p})$ where $f_{0}(t;z)$ is defined in the proof of (i).
This completes the proof.
\end{proof}


\begin{proof}[Proof of Proposition~\ref{prop:bounds}]
One can compute that
\begin{align}
\Lambda_{\rho}[g,h]
&=\frac{1}{2(1-\rho^{2})}\int_{0}^{1}\left(\frac{g'(t)}{\eta(e^{g(t)})\sqrt{e^{h(t)}}}\right)^{2}dt
+\frac{1}{2(1-\rho^{2})}\int_{0}^{1}\left(\frac{h'(t)}{\sigma(e^{h(t)})}\right)^{2}dt
\nonumber
\\
&\qquad\qquad\qquad
-\frac{2\rho}{2(1-\rho^{2})}\int_{0}^{1}\frac{g'(t)}{\eta(e^{g(t)})\sqrt{e^{h(t)}}}\frac{h'(t)}{\sigma(e^{h(t)})}dt
\nonumber
\\
&\leq
\frac{1}{2(1-\rho^{2})}\int_{0}^{1}\left(\frac{g'(t)}{\eta(e^{g(t)})\sqrt{e^{h(t)}}}\right)^{2}dt
+\frac{1}{2(1-\rho^{2})}\int_{0}^{1}\left(\frac{h'(t)}{\sigma(e^{h(t)})}\right)^{2}dt
\nonumber
\\
&\qquad\qquad\qquad
+\frac{|\rho|}{2(1-\rho^{2})}\int_{0}^{1}\left[\left(\frac{g'(t)}{\eta(e^{g(t)})\sqrt{e^{h(t)}}}\right)^{2}+\left(\frac{h'(t)}{\sigma(e^{h(t)})}\right)^{2}\right]dt
\nonumber
\\
&=\frac{1+|\rho|}{2(1-\rho^{2})}\int_{0}^{1}\left(\frac{g'(t)}{\eta(e^{g(t)})\sqrt{e^{h(t)}}}\right)^{2}dt
+\frac{1+|\rho|}{2(1-\rho^{2})}\int_{0}^{1}\left(\frac{h'(t)}{\sigma(e^{h(t)})}\right)^{2}dt
\nonumber
\\
&=\frac{1+|\rho|}{1-\rho^{2}}\Lambda_{0}[g,h].
\end{align}
Hence, we conclude that
\begin{equation}
\mathcal{I}_{\rho}(S_{0},V_{0},K)
\leq\frac{1+|\rho|}{1-\rho^{2}}\mathcal{I}_{0}(S_{0},V_{0},K)
=\frac{1}{1-|\rho|}\mathcal{I}_{0}(S_{0},V_{0},K),
\end{equation}
where $\mathcal{I}_{0}(S_{0},V_{0},K)$ is computed in closed form in Proposition~\ref{prop:variational}.

Similarly, one can compute that
\begin{align}
\Lambda_{\rho}[g,h]
&\geq
\frac{1}{2(1-\rho^{2})}\int_{0}^{1}\left(\frac{g'(t)}{\eta(e^{g(t)})\sqrt{e^{h(t)}}}\right)^{2}dt
+\frac{1}{2(1-\rho^{2})}\int_{0}^{1}\left(\frac{h'(t)}{\sigma(e^{h(t)})}\right)^{2}dt
\nonumber
\\
&\qquad\qquad\qquad
-\frac{|\rho|}{2(1-\rho^{2})}\int_{0}^{1}\left[\left(\frac{g'(t)}{\eta(e^{g(t)})\sqrt{e^{h(t)}}}\right)^{2}+\left(\frac{h'(t)}{\sigma(e^{h(t)})}\right)^{2}\right]dt
\nonumber
\\
&=\frac{1-|\rho|}{1-\rho^{2}}\Lambda_{0}[g,h],
\end{align}
which implies that
\begin{equation}
\mathcal{I}_{\rho}(S_{0},V_{0},K)
\geq\frac{1-|\rho|}{1-\rho^{2}}\mathcal{I}_{0}(S_{0},V_{0},K)
=\frac{1}{1+|\rho|}\mathcal{I}_{0}(S_{0},V_{0},K),
\end{equation}
where $\mathcal{I}_{0}(S_{0},V_{0},K)$ is computed in closed form in Proposition~\ref{prop:variational}.
This completes the proof.
\end{proof}


\begin{proof}[Proof of Proposition~\ref{prop:bounds:2}]
We recall from \eqref{I:rho:1} that 
\begin{align}
\mathcal{I}_{\rho}(S_{0},V_{0},K)
=\inf_{\substack{g(0)=\log S_{0},h(0)=\log V_{0}\\
\int_{0}^{1}e^{h(t)}\eta^{2}(e^{g(t)})dt=K}}\Lambda_{\rho}[g,h],
\end{align}
where $\Lambda_{\rho}[g,h]$ is defined in \eqref{I:rho:2}.
Since $g_{0},h_{0}$ are the optimal solutions for the variational problem \eqref{I:rho:1} 
when $\rho=0$, they satisfy the constraints $g_{0}(0)=\log S_{0}$, $h_{0}(0)=\log V_{0}$
and $\int_{0}^{1}e^{h_{0}(t)}\eta^{2}(e^{g_{0}(t)})dt=K$.
Therefore, we conclude that
\begin{align}
\mathcal{I}_{\rho}(S_{0},V_{0},K)
\leq\Lambda_{\rho}[g_{0},h_{0}],
\end{align}
and this completes the proof.
\end{proof}


\begin{proof}[Proof of Proposition~\ref{prop:bounds:3}]
By letting $g'(t)\equiv 0$ in \eqref{I:rho:1}, we get $g(t)=\log S_{0}$ for every $0\leq t\leq 1$
and
\begin{equation}
\mathcal{I}_{\rho}(S_{0},V_{0},K)
\leq
\frac{1}{2(1-\rho^{2})}\inf_{h(0)=\log V_{0},\int_{0}^{1}e^{h(t)}\eta^{2}(S_{0})dt=K}
\int_{0}^{1}\left(\frac{h'(t)}{\sigma(e^{h(t)})}\right)^{2}dt.
\end{equation}
By using the definition of $\mathcal{J}(\cdot,\cdot)$ in \eqref{J:formula}, 
we complete the proof.
\end{proof}


\begin{proof}[Proof of Proposition~\ref{prop:pm:1}]

We proceed in a similar way to the proof of Theorem \ref{thm:OTM}.
Under Assumptions~\ref{assump:bounded} and \ref{assump:LDP}, by the sample-path large deviations for small time diffusions \cite{Varadhan} and \cite{Robertson2010}),
one can see that $\mathbb{Q}(\{(\log S_{tT},\log V_{tT}),0\leq t\leq 1\}\in\cdot)$ in the $\rho= \pm 1$ limit
satisfies a sample-path large deviation principle with the rate function
\begin{equation}
J[g,h] = \frac{1}{2}\int_{0}^{1} \left( \frac{g'(t)}{\eta(e^{g(t)})\sqrt{e^{h(t)}}}\right)^{2}dt, 
\end{equation}
with $g(0)=\log S_{0}$, $h(0)=\log V_{0}$ and $g,h$ being absolutely continuous, and since both $S_t,V_t$ are driven by a common Brownian motion, 
$g(t)$ and $h(t)$ are related as
\begin{equation}\label{gh:pm:1}
\frac{g'(t)}{\eta(e^{g(t)})\sqrt{e^{h(t)}}} =
\pm \frac{h'(t)}{\sigma(e^{h(t)})} \,,
\quad 0\leq t\leq 1,
\end{equation}
where the signs correspond to $\rho = \pm 1$, 
and the rate function is $+\infty$ otherwise. 

By an application of the contraction principle we get
\begin{align}
\lim_{T\rightarrow 0}T\log\mathbb{Q}\left(\frac{1}{T}\int_{0}^{T}V_{s}\eta^{2}(S_{s})ds\geq K\right)
=-\inf_{\substack{g(0)=\log S_{0},h(0)=\log V_{0}\\
\int_{0}^{1}e^{h(t)}\eta^{2}(e^{g(t)})dt=K}}
\Bigg\{\frac{1}{2}\int_{0}^{1}\left(\frac{h'(t)}{\sigma(e^{h(t)})}\right)^{2}dt\Bigg\},
\end{align}
where the infimum is also subject to the constraint \eqref{gh:pm:1}.
Hence, we conclude that the rate function for $\rho =\pm 1$ is
\begin{align}
& \mathcal{I}_{\pm 1}(S_0,V_0,K) =-\inf_{\substack{g(0)=\log S_{0},h(0)=\log V_{0}\\
\int_{0}^{1}e^{h(t)}\eta^{2}(e^{g(t)})dt=K}}
\Bigg\{\frac{1}{2}\int_{0}^{1}\left(\frac{h'(t)}{\sigma(e^{h(t)})}\right)^{2}dt\Bigg\},
\end{align}
where the infimum is also subject to the constraint \eqref{gh:pm:1}.
The variational problem can be simplified by eliminating $g(t)$ using \eqref{gh:pm:1}.

This relation implies that
\begin{equation}
\int_{0}^{t}\frac{g'(s)}{\eta(e^{g(s)})}ds=\int_{0}^{t}\frac{\pm \sqrt{e^{h(s)}}h'(s)}{\sigma(e^{h(s)})}ds,\qquad 0\leq t\leq 1,
\end{equation}
which is equivalent to
\begin{equation}
\int_{\log S_{0}}^{g(t)}\frac{dx}{\eta(e^{x})}=\int_{\log V_{0}}^{h(t)}\frac{\pm \sqrt{e^{x}}dx}{\sigma(e^{x})},\qquad 0\leq t\leq 1,
\end{equation}
where we used the constraints $g(0)=\log S_{0}$ and $h(0)=\log V_{0}$.
We can further compute that this is equivalent to
\begin{equation}
\int_{\log S_{0}}^{g(t)}\frac{e^{x}dx}{e^{x}\eta(e^{x})}=\int_{\log V_{0}}^{h(t)}\frac{\pm e^{x}dx}{\sqrt{e^{x}}\sigma(e^{x})},\qquad 0\leq t\leq 1,
\end{equation}
which is equivalent to
\begin{equation}
\int_{S_{0}}^{e^{g(t)}}\frac{dx}{x\eta(x)}=\int_{V_{0}}^{e^{h(t)}}\frac{\pm dx}{\sqrt{x}\sigma(x)},\qquad 0\leq t\leq 1.
\end{equation}
Therefore, given $h$, the optimal $g$ is given by
\begin{equation}
e^{g(t)}=\mathcal{F}_{\pm}(e^{h(t)}),\qquad 0\leq t\leq 1,
\end{equation}
where $\mathcal{F}_{\pm}(\cdot)$ is defined as:
\begin{equation}
\int_{S_{0}}^{\mathcal{F}_{\pm}(x)}\frac{dy}{y\eta(y)}=\int_{V_{0}}^{x}\frac{\pm dy}{\sqrt{y}\sigma(y)},
\end{equation}
for any $x>0$.
Hence, we conclude that
\begin{equation}
\mathcal{I}_{\pm 1}(S_{0},V_{0},K)
=\inf_{\substack{h(0)=\log V_{0}\\
\int_{0}^{1}e^{h(t)}\eta^{2}(\mathcal{F}_{\pm}(e^{h(t)}))dt=K}}\frac{1}{2}\int_{0}^{1}\left(\frac{h'(t)}{\sigma(e^{h(t)})}\right)^{2}dt.
\end{equation}
This completes the proof.
\end{proof}


\begin{proof}[Proof of Proposition~\ref{prop:rho}]
We recall from \eqref{I:rho:1}-\eqref{I:rho:2} that
\begin{align}
\mathcal{I}_{\rho}(S_{0},V_{0},K)
=\inf_{\substack{g(0)=\log S_{0},h(0)=\log V_{0}\\
\int_{0}^{1}e^{h(t)}\eta^{2}(e^{g(t)})dt=K}}\Lambda_{\rho}[g,h],
\end{align}
where
\begin{align}
\Lambda_{\rho}[g,h]:=\frac{1}{2(1-\rho^{2})}\int_{0}^{1}\left(\frac{g'(t)}{\eta(e^{g(t)})\sqrt{e^{h(t)}}}-\frac{\rho h'(t)}{\sigma(e^{h(t)})}\right)^{2}dt
+\frac{1}{2}\int_{0}^{1}\left(\frac{h'(t)}{\sigma(e^{h(t)})}\right)^{2}dt.
\end{align}
In particular, 
\begin{equation}
\mathcal{I}_{\rho}(S_{0},V_{0},K)
\leq\frac{1}{2}\int_{0}^{1}\left(\frac{h'(t)}{\sigma(e^{h(t)})}\right)^{2}dt,
\end{equation}
for any $g,h$ that satisfies: $g(0)=\log S_{0}$, $h(0)=\log V_{0}$, 
$\int_{0}^{1}e^{h(t)}\eta^{2}(e^{g(t)})dt=K$ and
\begin{equation}
\frac{g'(t)}{\eta(e^{g(t)})\sqrt{e^{h(t)}}}=\frac{\rho h'(t)}{\sigma(e^{h(t)})},\qquad 0\leq t\leq 1.
\end{equation}
By following the proof of Proposition~\ref{prop:pm:1}, we can compute
that 
\begin{equation}
e^{g(t)}=\mathcal{F}_{\rho}(e^{h(t)}),\qquad 0\leq t\leq 1,
\end{equation}
where $\mathcal{F}_{\rho}$ is defined as:
\begin{equation}
\int_{S_{0}}^{\mathcal{F}_{\rho}(x)}\frac{dy}{y\eta(y)}=\int_{V_{0}}^{x}\frac{\rho dy}{\sqrt{y}\sigma(y)},\qquad\text{for any $x>0$}.
\end{equation}
Since this holds for any $h$ that satisfies the constraints 
$\int_{0}^{1}e^{h(t)}\eta^{2}(\mathcal{F}_{\rho}(e^{h(t)}))dt=K$ and $h(0)=\log V_{0}$, by taking the infimum over $h$, the proof is complete.
\end{proof}


\begin{proof}[Proof of Proposition~\ref{prop:first:order}]
We start with an expansion for the functions $g,h$ in powers of $x:=\log\left(\frac{K}{V_{0}\eta_0^{2}}\right)$ of the form
\begin{equation}\label{g:h:expansion}
g(t)=g_{0}(t)+xg_{1}(t)+O(x^{2}),
\qquad
h(t)=h_{0}(t)+xh_{1}(t)+O(x^{2}),
\end{equation}
We expand also the Lagrange multiplier as $\lambda=\lambda_{0}+x\lambda_{1}+O(x^{2})$ as $x\rightarrow 0$.

The zero-th order terms in these expansions are
$g_{0}(t)\equiv\log S_{0}$, $h_{0}(t)\equiv\log V_{0}$
such that $g'_{0}(t)\equiv 0$ and $h'_{0}(t)\equiv 0$.
The transversality conditions $g'(1) = h'(1) = 0$ are satisfied if and only if we have 
$g'_k(1) = h'_k(1) = 0$ for all $k\geq 1$. Also, the boundary conditions $g(0) = \log S_0, h(0)=\log V_0$ imply that one must have $g_k(0) = h_k(0)=0$ for all $k\geq 1$.


The constraint $\int_0^1 e^{h(t)}\eta^{2}(e^{g(t)})dt = K$ relates $g_k(t),h_{k}(t)$ to all $g_j(t),h_j(t)$ of lower order $0 \leq j < k$.
This is written equivalently as
\begin{equation}
\int_0^1 V_{0}e^{x h_1(t) + x^2 h_2(t) + O(x^{3}) }\eta^{2}\left(S_{0}e^{x g_1(t) + x^2 g_2(t) + O(x^{3}) }\right)dt = V_0\eta^{2}(S_0) e^x, 
\end{equation}
as $x\rightarrow 0$.
Expanding in $x$ and selecting terms of the same power of $x$ on both sides gives the constraints
\begin{align}
O(x): & \int_0^1 h_1(t) dt+ 2\tilde \eta_{1} \int_{0}^{1}g_{1}(t)dt = 1\,,\label{norm} \\
O(x^2): &  \int_0^1 \left( h_2(t) + 2\tilde \eta_1 g_2(t) + \frac12 (h_1(t))^2 + (2\tilde \eta_2 + \tilde \eta_1^2) (g_1(t))^2 + 2\tilde \eta_1 g_1(t) h_1(t) \right) dt  = \frac12\,,\label{norm:2}
\end{align}
and so on, where we denoted $\tilde \eta_k := \eta_k/\eta_0$.

We substitute the expansions \eqref{g:h:expansion} into the Euler-Lagrange equations 
\eqref{EL:1}-\eqref{EL:2} and expand in $x$. Let us consider the terms of given order in $x$ resulting from this expansion.

\textbf{Order $O(x^0)$.}
At order $O(x^0)$, the equation \eqref{EL:1} gives
$\lambda_0 V_0 = 0$ which gives $\lambda_0=0$.
Both sides of the equation \eqref{EL:2}  vanish identically at this order.

\textbf{Order $O(x)$.}
At order $O(x)$, the two equations become
\begin{align}
\frac{d}{dt}\left(\frac{1}{1-\rho^{2}}\frac{g'_{1}(t)}{\eta^{2}_0 V_{0}}
-\frac{\rho}{1-\rho^{2}}\frac{h'_{1}(t)}{\eta_0 \sqrt{V_{0}}\sigma_0}\right)
=2\lambda_{1}V_{0}\eta_{0}\eta_{1},\label{EL:1:first}
\end{align}
and
\begin{align}
\frac{d}{dt}\left(\frac{1}{1-\rho^{2}}\frac{h'_{1}(t)}{\sigma^{2}_0}
-\frac{\rho}{1-\rho^{2}}\frac{g'_{1}(t)}{\eta_0 \sqrt{V_{0}}\sigma_0 }\right)=\lambda_{1}V_{0}\eta_{0}^{2},\label{EL:2:first}
\end{align}
with the constraints that $g_{1}(0)=0$, $h_{1}(0)=0$ and $\int_0^1 h_1(t) dt+\frac{2\eta_{1}}{\eta_{0}}\int_{0}^{1}g_{1}(t)dt = 1$,
and the transversality condition gives $g_{1}'(1)=h_{1}'(1)=0$.
We can re-write \eqref{EL:1:first}-\eqref{EL:2:first} as
\begin{align}
&\frac{g''_{1}(t)}{\eta^{2}_0V_{0}}-\frac{\rho h''_{1}(t)}{\eta_0\sqrt{V_{0}}\sigma_0}=2\lambda_{1}V_{0}\eta_{0}\eta_{1}(1-\rho^{2}),
\\
&\frac{h''_{1}(t)}{\sigma^{2}_0}-\frac{\rho g''_{1}(t)}{\eta_0\sqrt{V_{0}}\sigma_0}=\lambda_{1}V_{0}\eta_{0}^{2}(1-\rho^{2}),
\end{align}
which implies that
\begin{align}
&g''_{1}(t)=\lambda_{1}\left(2V_{0}^{2}\eta_{0}^{3}\eta_{1}+V_{0}^{3/2}\eta_{0}^{3}\rho\sigma_{0}\right),\qquad\qquad g_{1}(0)=g'_{1}(1)=0,\label{g:1:twice}
\\
&h''_{1}(t)=\lambda_{1}\left(\sigma_{0}^{2}V_{0}\eta_{0}^{2}+2\rho\sigma_{0}V_{0}^{3/2}\eta_{0}^{2}\eta_{1}\right),\qquad\qquad h_{1}(0)=h'_{1}(1)=0,\label{h:1:twice}
\end{align}
with $\int_0^1 h_1(t) dt+\frac{2\eta_{1}}{\eta_{0}}\int_{0}^{1}g_{1}(t)dt = 1$. 
The equations \eqref{g:1:twice}-\eqref{h:1:twice} can be integrated using the boundary condition $g'_1(1)=h'(1)=0$ to give
\begin{align}
&g'_{1}(t)= 
\lambda_{1}\left(2\sqrt{V_{0}}\eta_{1}+\rho\sigma_{0}\right) V_0^{3/2} \eta_0^3
(t-1),\label{g:1:once}
\\
&h'_{1}(t)= \lambda_{1}\left(\sigma_{0}+2\rho\sqrt{V_{0}}\eta_{1}\right)\sigma_0 V_0 \eta_0^2 (t-1)  \,.\label{h:1:once}
\end{align}
Integrating \eqref{g:1:once}-\eqref{h:1:once} again using the boundary conditions $g_1(0)=h_1(0)=0$ gives
\begin{align}
&g_{1}(t)=\frac{1}{2} \lambda_{1}\left(2\sqrt{V_{0}}\eta_{1}+\rho\sigma_{0}\right)V_0^{3/2} \eta_0^3 (t^{2}-2t),
\\
&h_{1}(t)=\frac{1}{2}\lambda_{1}\left(\sigma_{0}+2\rho\sqrt{V_{0}}\eta_{1}\right) \sigma_0 V_0 \eta_0^2 (t^{2}-2t).
\end{align}

The constant $\lambda_1$ is determined from the first normalization condition (\ref{norm})
\begin{equation}
\int_0^1 h_1(t) dt+\frac{2\eta_{1}}{\eta_{0}}\int_{0}^{1}g_{1}(t)dt =-\frac{1}{2}\lambda_{1}\left(\sigma_{0}^{2}V_{0}\eta_{0}^{2}+4\rho\sigma_{0}V_{0}^{3/2}\eta_{0}^{2}\eta_{1}+4V_{0}^{2}\eta_{0}^{2}\eta_{1}^{2}\right)\frac{2}{3}=1,
\end{equation}
which gives $\lambda_{1}=-\frac{3}{V_0 \eta_0^2 (\sigma_{0}^{2}+4\rho\sigma_{0}V_{0}^{1/2}\eta_{1}+4V_{0}\eta_{1}^{2})}$. We conclude that
\begin{align}
&g_{1}(t)=\frac{3}{2} \frac{(2V_{0}^{1/2}\eta_{1}+\rho\sigma_{0})V_0^{1/2} \eta_0}{\sigma_{0}^{2}+4\rho\sigma_{0}V_{0}^{1/2}\eta_{1}+4V_{0}\eta_{1}^{2}}(2t-t^{2}),\label{g:1:function}
\\
&h_{1}(t)=\frac{3}{2}\frac{\sigma_{0}^{2}+2\rho\sigma_{0}V_{0}^{1/2}\eta_{1}}{\sigma_{0}^{2}+4\rho\sigma_{0}V_{0}^{1/2}\eta_{1}+4V_{0}\eta_{1}^{2}}(2t-t^{2}).\label{h:1:function}
\end{align}

Finally, by plugging \eqref{g:h:expansion} into \eqref{I:rho:2}, it follows from \eqref{I:rho:1}
that
\begin{align}
&\mathcal{I}_{\rho}(S_{0},V_{0},V_{0}\eta^{2}(S_{0})e^{x})
\nonumber
\\
&=\frac{x^{2}}{2(1-\rho^{2})}\int_{0}^{1}\left(\frac{g'_{1}(t)}{\eta(S_{0})\sqrt{V_{0}}}-\frac{\rho h'_{1}(t)}{\sigma_0}\right)^{2}dt
+\frac{x^{2}}{2}\int_{0}^{1}\left(\frac{h'_{1}(t)}{\sigma_0}\right)^{2}dt
+O(x^{3})
\nonumber
\\
&=\frac{x^{2}\lambda_{1}^{2}}{2(1-\rho^{2})}\int_{0}^{1}
\left(\frac{2V_{0}^{2}\eta_{0}^{3}\eta_{1}+V_{0}^{3/2}\eta_{0}^{3}\rho\sigma_{0}}{\eta(S_{0})\sqrt{V_{0}}}
-\frac{\rho}{\sigma_{0}}\left(\sigma_{0}^{2}V_{0}\eta_{0}^{2}+2\rho\sigma_{0}V_{0}^{3/2}\eta_{0}^{2}\eta_{1}\right)\right)^{2}(1-t)^{2}dt
\nonumber
\\
&\qquad
+\frac{x^{2}\lambda_{1}^{2}}{2}\int_{0}^{1}\left(\frac{\sigma_{0}^{2}V_{0}\eta_{0}^{2}+2\rho\sigma_{0}V_{0}^{3/2}\eta_{0}^{2}\eta_{1}}{\sigma_{0}}\right)^{2}(1-t)^{2}dt
+O(x^{3})
\nonumber
\\
&=\frac{x^{2}\lambda_{1}^{2}}{6(1-\rho^{2})}
\left(2V_{0}^{3/2}\eta_{0}^{2}\eta_{1}-2\rho^{2}V_{0}^{3/2}\eta_{0}^{2}\eta_{1}\right)^{2}
+\frac{x^{2}\lambda_{1}^{2}}{6}\left(\sigma_{0}V_{0}\eta_{0}^{2}+2\rho V_{0}^{3/2}\eta_{0}^{2}\eta_{1}\right)^{2}
+O(x^{3})
\nonumber
\\
&=\frac{x^{2}\lambda_{1}^{2}}{6}\left(\sigma_{0}^{2}V_{0}^{2}\eta_{0}^{4}+4\rho\sigma_{0}V_{0}^{5/2}\eta_{0}^{4}\eta_{1}+4V_{0}^{3}\eta_{0}^{4}\eta_{1}^{2}\right)
+O(x^{3})
\nonumber
\\
&=\frac{3x^{2}}{2}\frac{\sigma_{0}^{2}V_{0}^{2}\eta_{0}^{4}+4\rho\sigma_{0}V_{0}^{5/2}\eta_{0}^{4}\eta_{1}+4V_{0}^{3}\eta_{0}^{4}\eta_{1}^{2}}{\left(\sigma_{0}^{2}V_{0}\eta_{0}^{2}+4\rho\sigma_{0}V_{0}^{3/2}\eta_{0}^{2}\eta_{1}+4V_{0}^{2}\eta_{0}^{2}\eta_{1}^{2}\right)^{2}}+O(x^{3})
\nonumber
\\
&=\frac{3}{2\left(\sigma_{0}^{2}+4\rho\sigma_{0}V_{0}^{1/2}\eta_{1}+4V_{0}\eta_{1}^{2}\right)}x^{2}+O(x^{3}),
\end{align}
as $x\rightarrow 0$.
This completes the proof for the $O(x^2)$ term in the rate function.
\vspace{0.5cm}

\textbf{Order $O(x^2)$.} The $O(x^2)$ terms in the expansion of the Euler-Lagrange equations give a system of two linear equations in $g''_2(t), h''_2(t)$. They are solved with the solutions
\begin{align}\label{g2pp}
g''_2(t) &= a + b t + c t^2 \,,\qquad
h''_2(t) = \bar a + \bar b t + \bar c t^2 \,,
\end{align}
with coefficients $a,b,c$ and $\bar a, \bar b, \bar c$ which depend on known parameters $\sigma_i, \eta_j$ and the unknown $\lambda_2$. 

Taking into account the boundary conditions at $t=1$,
the equations (\ref{g2pp}) can be integrated as
\begin{align}
    g'_2(t) &= (t-1) \left( a + \frac12 b (t+1) + \frac13 c(t^2 + t + 1) \right)\,,\label{int:again:1}\\
    h'_2(t) &= (t-1) \left( \bar a + \frac12 \bar b (t+1) + \frac13 \bar c(t^2 + t + 1) \right)\,.\label{int:again:2}
\end{align}
Integrating again \eqref{int:again:1}-\eqref{int:again:2} using the boundary condition at $t=0$ gives
\begin{align}\label{g2sol}
    g_2(t) &= \frac12 a t(t-2) + \frac16 b t (t^2-3) + \frac{1}{12} c t (t^3-4)\,,\\
    \label{h2sol}
    h_2(t) &= \frac12 \bar a t(t-2) + \frac16 \bar b t (t^2-3) + \frac{1}{12} \bar c t (t^3-4) \,.
\end{align}

The coefficient $\lambda_2$ is determined from the normalization condition \eqref{norm:2} which can be re-written as
\begin{align}
& \int_0^1 \left(h_2(t) + 2\tilde \eta_1 g_2(t) \right) dt  \nonumber \\
& = \frac12 - \int_0^1 
\left(\frac12 (h_1(t))^2 + (\tilde \eta_1^2 + 2\tilde \eta_2) (g_1(t))^2 + 2\tilde \eta_1 g_1(t) h_1(t) \right) dt\,.\label{normalization:condition}
\end{align}
The integral on the right-hand side of \eqref{normalization:condition} depends only on the functions $g_1(t),h_1(t)$ which have been determined in the previous step; see \eqref{g:1:function}-\eqref{h:1:function}. Thus, this condition introduces a linear constraint on the coefficients $a,b,c,\bar a,\bar b,\bar c$ which is used to solve for $\lambda_2$.

The final expressions for the coefficients in (\ref{g2pp}) are rather lengthy so we give them in Appendix~\ref{sec:g2:h2}. 

Substituting the expressions for $g'_2(t),h'_2(t)$ from \eqref{int:again:1}-\eqref{int:again:2} into the $O(x^3)$ term in the expansion of the rate function \eqref{I:rho:lambda}
yields the stated result. 
\end{proof}


\begin{proof}[Proof of Theorem~\ref{thm:ATM}]
We only provide the proof for the ATM call option. 
The case for the put option is similar.
First of all, it is easy to see that
\begin{equation}
\left|C(T)-\mathbb{E}\left[\left(\frac{1}{T}\int_{0}^{T}V_{t}\eta^{2}(S_{t})dt-K\right)^{+}\right]\right|
=O(T),
\end{equation}
as $T\rightarrow 0$, which follows from the estimate:
\begin{align}
&\left|C(T)-\mathbb{E}\left[\left(\frac{1}{T}\int_{0}^{T}V_{t}\eta^{2}(S_{t})dt-K\right)^{+}\right]\right|
\nonumber
\\
&\leq
\left|1-e^{-rT}\right|\cdot\left|\mathbb{E}\left[\left(\frac{1}{T}\int_{0}^{T}V_{t}\eta^{2}(S_{t})dt-K\right)^{+}\right]\right|
\nonumber
\\
&\leq\left|1-e^{-rT}\right|\cdot\left(\frac{1}{T}\int_{0}^{T}\mathbb{E}[V_{t}\eta^{2}(S_{t})]dt+K\right)
\nonumber
\\
&\leq\left|1-e^{-rT}\right|\cdot\left(\frac{1}{T}\int_{0}^{T}V_{0}e^{(M_{\mu}+\frac{1}{2}M_{\sigma}^{2})t}M_{\eta}^{2}dt+K\right),
\end{align}
where we used the fact that $|\eta(\cdot)|\leq M_{\eta}$
and the bound $\mathbb{E}[V_{t}]\leq V_{0}e^{M_{\mu}t+\frac{1}{2}M_{\sigma}^{2}t}$ for every $t\geq 0$ from \eqref{V:p:bound}.

It has been shown in Theorem 3.2. \cite{PWZ2024} that 
under Assumptions~\ref{assump:bounded}, \ref{assump:lip} 
and $\max_{0\leq t\leq T}\mathbb{E}[(S_{t})^{4}]=O(1)$
as $T\rightarrow 0$,
we have that uniformly in $0\leq t\leq T$,
\begin{equation}\label{estimate:VIX:paper}
\mathbb{E}\left[\left|S_{t}-\hat{S}_{t}\right|^{2}\right]=O\left(T^{3/2}\right),
\qquad
\mathbb{E}\left[\left|V_{t}-\hat{V}_{t}\right|\right]=O(T),
\end{equation}
where
\begin{align}
&\hat{S}_{t}=S_{0}+S_{0}\eta(S_{0})\sqrt{V_{0}}\left(\sqrt{1-\rho^{2}}W_{t}+\rho Z_{t}\right),
\\
&\hat{V}_{t}=V_{0}+V_{0}\sigma(V_{0})Z_{t}.
\end{align}

Since $x\mapsto x^{+}$ is $1$-Lipschitz, we can compute that
\begin{align}
&\left|\mathbb{E}\left[\left(\frac{1}{T}\int_{0}^{T}\hat{V}_{s}\eta^{2}(\hat{S}_{s})ds-K\right)^{+}\right]
-\mathbb{E}\left[\left(\frac{1}{T}\int_{0}^{T}V_{s}\eta^{2}(S_{s})ds-K\right)^{+}\right]\right|
\nonumber
\\
&\leq
\mathbb{E}\left|\frac{1}{T}\int_{0}^{T}\hat{V}_{s}\eta^{2}(\hat{S}_{s})ds-\frac{1}{T}\int_{0}^{T}V_{s}\eta^{2}(S_{s})ds\right|
\nonumber
\\
&\leq
\frac{1}{T}\int_{0}^{T}\mathbb{E}\left|\hat{V}_{s}\eta^{2}(\hat{S}_{s})-V_{s}\eta^{2}(S_{s})\right|ds.
\end{align}
Moreover, uniformly in $0\leq s\leq T$, 
\begin{align}
\mathbb{E}\left|\hat{V}_{s}\eta^{2}(\hat{S}_{s})-V_{s}\eta^{2}(S_{s})\right|
\leq
\mathbb{E}\left|\hat{V}_{s}\eta^{2}(\hat{S}_{s})-\hat{V}_{s}\eta^{2}(S_{s})\right|
+\mathbb{E}\left|\hat{V}_{s}\eta^{2}(S_{s})-V_{s}\eta^{2}(S_{s})\right|.
\end{align}
Under the assumption that $|\eta(\cdot)|\leq M_{\eta}$ uniformly
and $\eta(\cdot)$ is $\ell_{\eta}$-Lipschitz, we have
\begin{align}
\mathbb{E}\left|\hat{V}_{s}\eta^{2}(\hat{S}_{s})-V_{s}\eta^{2}(S_{s})\right|
&\leq
2M_{\eta}\ell_{\eta}\mathbb{E}\left[|\hat{V}_{s}|\cdot\left|\hat{S}_{s}-S_{s}\right|\right]
+M_{\eta}^{2}\mathbb{E}\left|\hat{V}_{s}-V_{s}\right|
\nonumber
\\
&\leq
2M_{\eta}\ell_{\eta}\left(\mathbb{E}\left[|\hat{V}_{s}|^{2}\right]\right)^{1/2}
\left(\mathbb{E}\left[\left|\hat{S}_{s}-S_{s}\right|^{2}\right]\right)^{1/2}
+M_{\eta}^{2}\mathbb{E}\left|\hat{V}_{s}-V_{s}\right|
\nonumber
\\
&=O(T^{3/4}),
\end{align}
as $T\rightarrow 0$ by applying \eqref{estimate:VIX:paper}
and the simple bound $\mathbb{E}|\hat{V}_{s}|^{2}\leq 2V_{0}^{2}+2V_{0}^{2}\sigma^{2}(V_{0})T$
uniformly in $0\leq t\leq T$.

Therefore, the call option can be approximated by (with $K=V_{0}\eta^{2}(S_{0})$):
\begin{equation}
\left|\mathbb{E}\left[\left(\frac{1}{T}\int_{0}^{T}\hat{V}_{s}\eta^{2}(\hat{S}_{s})ds-K\right)^{+}\right]
-\mathbb{E}\left[\left(\frac{1}{T}\int_{0}^{T}V_{s}\eta^{2}(S_{s})ds-K\right)^{+}\right]\right|
=O(T^{3/4}),
\end{equation}
as $T\rightarrow 0$.

Next, we can compute that
\begin{align}
&\mathbb{E}\left[\left(\frac{1}{T}\int_{0}^{T}\hat{V}_{s}\eta^{2}(\hat{S}_{s})ds-V_{0}\eta^{2}(S_{0})\right)^{+}\right]
\nonumber
\\
&=\mathbb{E}\Bigg[\Bigg(\frac{1}{T}\int_{0}^{T}\left(V_{0}+V_{0}\sigma(V_{0})Z_{s}\right)
\nonumber
\\
&\qquad\qquad\qquad
\cdot\eta^{2}\left(S_{0}+S_{0}\eta(S_{0})\sqrt{V_{0}}\left(\sqrt{1-\rho^{2}}W_{s}+\rho Z_{s}\right)\right)ds
-V_{0}\eta^{2}(S_{0})\Bigg)^{+}\Bigg].
\end{align}
Since $x\mapsto x^{+}$ is $1$-Lipschitz, 
we have
\begin{align}
&\Bigg|\mathbb{E}\left[\left(\frac{1}{T}\int_{0}^{T}\left(V_{0}+V_{0}\sigma(V_{0})Z_{s}\right)
\eta^{2}\left(S_{0}+S_{0}\eta(S_{0})\sqrt{V_{0}}\left(\sqrt{1-\rho^{2}}W_{s}+\rho Z_{s}\right)\right)ds-V_{0}\eta^{2}(S_{0})\right)^{+}\right]
\nonumber
\\
&\qquad
-\mathbb{E}\Bigg[\Bigg(\frac{1}{T}\int_{0}^{T}\Bigg(V_{0}\eta^{2}(S_{0})
+V_{0}2\eta(S_{0})\eta'(S_{0})S_{0}\eta(S_{0})\sqrt{V_{0}}\left(\sqrt{1-\rho^{2}}W_{s}+\rho Z_{s}\right)
\nonumber
\\
&\qquad\qquad\qquad\qquad\qquad
+\eta^{2}(S_{0})V_{0}\sigma(V_{0})Z_{s}\Bigg)ds-V_{0}\eta^{2}(S_{0})\Bigg)^{+}\Bigg]\Bigg|
\nonumber
\\
&\leq
\mathbb{E}\Bigg[\frac{1}{T}\int_{0}^{T}\Big|\left(V_{0}+V_{0}\sigma(V_{0})Z_{s}\right)
\Big(\eta^{2}\left(S_{0}+S_{0}\eta(S_{0})\sqrt{V_{0}}\left(\sqrt{1-\rho^{2}}W_{s}+\rho Z_{s}\right)\right)
\nonumber
\\
&\qquad\qquad
-\eta^{2}(S_{0})-(\eta^{2}(S_{0}))'S_{0}\eta(S_{0})\sqrt{V_{0}}\left(\sqrt{1-\rho^{2}}W_{s}+\rho Z_{s}\right)\Big)
\Big|ds\Bigg]
\nonumber
\\
&\qquad
+\mathbb{E}\left[\frac{1}{T}\int_{0}^{T}\left|V_{0}\sigma(V_{0})Z_{s}(\eta^{2}(S_{0}))'S_{0}\eta(S_{0})\sqrt{V_{0}}\left(\sqrt{1-\rho^{2}}W_{s}+\rho Z_{s}\right)\right|ds\right].
\end{align}
Under the assumption that $\sup_{x\in\mathbb{R}^{+}}|(\eta^{2})''(x)|<\infty$, there exists some
constant $C>0$, such that
\begin{align}
&\mathbb{E}\Bigg[\frac{1}{T}\int_{0}^{T}\Big|\left(V_{0}+V_{0}\sigma(V_{0})Z_{s}\right)
\Big(\eta^{2}(S_{0}+S_{0}\eta(S_{0})\sqrt{V_{0}}\left(\sqrt{1-\rho^{2}}W_{s}+\rho Z_{s}\right))
\nonumber
\\
&\qquad\qquad
-\eta^{2}(S_{0})-(\eta^{2}(S_{0}))'S_{0}\eta(S_{0})\sqrt{V_{0}}\left(\sqrt{1-\rho^{2}}W_{s}+\rho Z_{s}\right)\Big)
\Big|ds\Bigg]
\nonumber
\\
&\leq
\frac{C}{T}\int_{0}^{T}\mathbb{E}\left[\left|\left(V_{0}+V_{0}\sigma(V_{0})Z_{s}\right)
\left(\sqrt{1-\rho^{2}}W_{s}+\rho Z_{s}\right)^{2}\right|\right]ds
\nonumber
\\
&\leq
\frac{C}{T}\int_{0}^{T}\left(\mathbb{E}\left[\left(V_{0}+V_{0}\sigma(V_{0})Z_{s}\right)^{2}\right]\right)^{1/2}
\left(\mathbb{E}\left[\left(\sqrt{1-\rho^{2}}W_{s}+\rho Z_{s}\right)^{4}\right]\right)^{1/2}ds
\nonumber
\\
&=\frac{C}{T}\int_{0}^{T}\left(V_{0}^{2}+V_{0}^{2}\sigma^{2}(V_{0})s\right)^{1/2}
\left(3s^{2}\right)^{1/2}ds=O(T),
\end{align}
as $T\rightarrow 0$.
In addition, we can compute that
\begin{align}
&\mathbb{E}\left[\frac{1}{T}\int_{0}^{T}\left|V_{0}\sigma(V_{0})Z_{s}(\eta^{2}(S_{0}))'S_{0}\eta(S_{0})\sqrt{V_{0}}\left(\sqrt{1-\rho^{2}}W_{s}+\rho Z_{s}\right)\right|ds\right]
\nonumber
\\
&=\left(V_{0}\sigma(V_{0})(\eta^{2}(S_{0}))'S_{0}\eta(S_{0})\sqrt{V_{0}}\right)^{2}
\frac{1}{T}\int_{0}^{T}\mathbb{E}\left[\left|Z_{s}\left(\sqrt{1-\rho^{2}}W_{s}+\rho Z_{s}\right)\right|\right]ds
\nonumber
\\
&\leq\left(V_{0}\sigma(V_{0})(\eta^{2}(S_{0}))'S_{0}\eta(S_{0})\sqrt{V_{0}}\right)^{2}
\nonumber
\\
&\qquad\qquad\cdot
\frac{1}{T}\int_{0}^{T}\left(\mathbb{E}\left[(Z_{s})^{2}\right]\right)^{1/2}
\left(\mathbb{E}\left[\left(\sqrt{1-\rho^{2}}W_{s}+\rho Z_{s}\right)^{2}\right]\right)^{1/2}ds
\nonumber
\\
&=\left(V_{0}\sigma(V_{0})(\eta^{2}(S_{0}))'S_{0}\eta(S_{0})\sqrt{V_{0}}\right)^{2}
\frac{1}{T}\int_{0}^{T}\sqrt{s}\sqrt{s}ds=O(T),
\end{align}
as $T\rightarrow 0$.

Finally, we can compute that
\begin{align}
&\mathbb{E}\Bigg[\Bigg(\frac{1}{T}\int_{0}^{T}\Bigg(V_{0}\eta^{2}(S_{0})
+V_{0}(\eta^{2}(S_{0}))'S_{0}\eta(S_{0})\sqrt{V_{0}}\left(\sqrt{1-\rho^{2}}W_{s}+\rho Z_{s}\right)
\nonumber
\\
&\qquad\qquad\qquad\qquad\qquad
+\eta^{2}(S_{0})V_{0}\sigma(V_{0})Z_{s}\Bigg)ds-V_{0}\eta^{2}(S_{0})\Bigg)^{+}\Bigg]
\nonumber
\\
&=\mathbb{E}\Bigg[\Bigg(2V_{0}\eta^{2}(S_{0})\eta'(S_{0})S_{0}\sqrt{V_{0}}\frac{1}{T}\int_{0}^{T}\left(\sqrt{1-\rho^{2}}W_{s}+\rho Z_{s}\right)ds
\nonumber
\\
&\qquad\qquad\qquad\qquad\qquad\qquad\qquad\qquad
+\eta^{2}(S_{0})V_{0}\sigma(V_{0})\frac{1}{T}\int_{0}^{T}Z_{s}ds\Bigg)^{+}\Bigg].
\end{align}

It is easy to check that
$\frac{1}{T}\int_{0}^{T}W_{s}ds$ 
and $\frac{1}{T}\int_{0}^{T}Z_{s}ds$ 
are independently and identically distributed as $\mathcal{N}\left(0,\frac{T}{3}\right)$ (see e.g. \cite{PZAsian}).
Therefore, we have
\begin{align}
&\mathbb{E}\left[\left(2V_{0}\eta^{2}(S_{0})\eta'(S_{0})S_{0}\sqrt{V_{0}}\frac{1}{T}\int_{0}^{T}\left(\sqrt{1-\rho^{2}}W_{s}+\rho Z_{s}\right)ds
+\eta^{2}(S_{0})V_{0}\sigma(V_{0})\frac{1}{T}\int_{0}^{T}Z_{s}ds\right)^{+}\right]
\nonumber
\\
&=\mathbb{E}\Bigg[\Bigg(2V_{0}\eta^{2}(S_{0})\eta'(S_{0})S_{0}\sqrt{V_{0}}\sqrt{1-\rho^{2}}\frac{1}{T}\int_{0}^{T}W_{s}ds
\nonumber
\\
&\qquad\qquad\qquad\qquad
+\left(2V_{0}\eta^{2}(S_{0})\eta'(S_{0})S_{0}\sqrt{V_{0}}\rho+\eta^{2}(S_{0})V_{0}\sigma(V_{0})\right)\frac{1}{T}\int_{0}^{T}Z_{s}ds\Bigg)^{+}\Bigg]
\nonumber
\\
&=\Bigg(\left(2V_{0}\eta^{2}(S_{0})\eta'(S_{0})S_{0}\sqrt{V_{0}}\sqrt{1-\rho^{2}}\right)^{2}
\nonumber
\\
&\qquad\qquad
+\left(2V_{0}\eta^{2}(S_{0})\eta'(S_{0})S_{0}\sqrt{V_{0}}\rho+\eta^{2}(S_{0})V_{0}\sigma(V_{0})\right)^{2}\Bigg)^{1/2}
\frac{\sqrt{T}}{\sqrt{3}}
\mathbb{E}\left[X^{+}\right]
\nonumber
\\
&=\Bigg(\left(2V_{0}\eta^{2}(S_{0})\eta'(S_{0})S_{0}\sqrt{V_{0}}\sqrt{1-\rho^{2}}\right)^{2}
\nonumber
\\
&\qquad\qquad
+\left(2V_{0}\eta^{2}(S_{0})\eta'(S_{0})S_{0}\sqrt{V_{0}}\rho+\eta^{2}(S_{0})V_{0}\sigma(V_{0})\right)^{2}
\Bigg)^{1/2}
\frac{\sqrt{T}}{\sqrt{3}}
\frac{1}{\sqrt{2\pi}},
\end{align}
where $X\sim\mathcal{N}(0,1)$.

Hence, we conclude that, for ATM case,
\begin{align}
&\lim_{T\rightarrow 0}\frac{1}{\sqrt{T}}
\mathbb{E}\left[\left(\frac{1}{T}\int_{0}^{T}V_{s}\eta^{2}(S_{s})ds-V_{0}\eta^{2}(S_{0})\right)^{+}\right]
\nonumber
\\
&=\frac{1}{\sqrt{6\pi}}\Bigg(\left(2V_{0}\eta^{2}(S_{0})\eta'(S_{0})S_{0}\sqrt{V_{0}}\sqrt{1-\rho^{2}}\right)^{2}
\nonumber
\\
&\qquad\qquad\qquad
+\left(2V_{0}\eta^{2}(S_{0})\eta'(S_{0})S_{0}\sqrt{V_{0}}\rho+\eta^{2}(S_{0})V_{0}\sigma(V_{0})\right)^{2}\Bigg)^{1/2}
\nonumber
\\
&=\frac{1}{\sqrt{6\pi}}\sqrt{4V_{0}^{3}\eta^{4}(S_{0})(\eta'(S_{0}))^{2}S_{0}^{2}
+\eta^{4}(S_{0})V_{0}^{2}\sigma^{2}(V_{0})
+4\rho\eta^{4}(S_{0})\eta'(S_{0})S_{0}V_{0}^{5/2}\sigma(V_{0})}.
\end{align}
This completes the proof.
\end{proof}

\section{Coefficients of $g_2(t),h_2(t)$}
\label{sec:g2:h2}

We give in this Appendix the coefficients appearing in the functions $g_2(t), h_2(t)$ relevant for the $O(x^3)$ term in the rate function.
The details of proof are given in the proof of Proposition~\ref{prop:first:order}.
Recall that these functions are given by
\begin{align}
    g_2(t) &= \frac12 a t(t-2) + \frac16 b t (t^2-3) + \frac{1}{12} c t (t^3-4)\,,\label{g:2:expression}\\
    h_2(t) &= \frac12 \bar a t(t-2) + \frac16 \bar b t (t^2-3) + \frac{1}{12} \bar c t (t^3-4) \,.\label{h:2:expression}
\end{align}

The coefficients of $g_2(t)$ in \eqref{g:2:expression} can be expressed as an expansion in the leading volatility-of-volatility coefficient $\sigma_0$ as follows:
\begin{align}\label{a:b:c}
a := \frac{1}{10 D^3} 3 \eta_0 \sqrt{V_0} \cdot \sum_{j=0}^5 c_j \sigma_0^j \,\quad
b := -\frac{9}{2 D^2} \sum_{j=0}^3 d_j \sigma_0^j \,,\quad
c := \frac{9}{4 D^2} \sum_{j=0}^3 d_j \sigma_0^j\,, 
\end{align}
where we denote $D := \sigma_0^2 + 4\eta_1 \sigma_0 \rho \sqrt{V_0} + 4\eta_1^2 V_0$ and the coefficients in \eqref{a:b:c} are
\begin{align}
c_0 &:=64 \eta_1^3 (14 \eta_1^2 + 9 \eta_0 \eta_2) V_0^{5/2}\,, \\
c_1 &:= 144 \eta_1^3 \rho^2 \sigma_1 V_0^{3/2} + 
  16 \eta_1^2 (149 \eta_1^2 + 54 \eta_0 \eta_2) \rho V_0^2\,, \\
c_2 &:= 144 \eta_1^2 \rho \sigma_1 V_0 + 72 \eta_1^2 \rho^3 \sigma_1 V_0 + 
  640 \eta_1^3 V_0^{3/2}+ 
  8 \eta_1 (227 \eta_1^2 + 54 \eta_0 \eta_2) \rho^2 V_0^{3/2}\,, \\
c_3 &:= 
36 \eta_1 \sigma_1 \sqrt{V_0} + 72 \eta_1 \rho^2 \sigma_1 \sqrt{V_0} + 
  864 \eta_1^2 \rho V_0 + 4 (91 \eta_1^2 + 18 \eta_0 \eta_2) \rho^3 V_0\,, \\
c_4 &:= 
  18 \rho \sigma_1 + 86 \eta_1 \sqrt{V_0} + 212 \eta_1 \rho^2 \sqrt{V_0}\,,   \\
c_5 &:= 28\rho\,, 
\end{align}
and
\begin{align}
d_0 &:= 8 \eta_0 \eta_1 V_0^2 (5 \eta_1^2 + 2 \eta_0 \eta_2)\,, \\
d_1 &:= 4 \eta_0 \eta_1 \rho^2 \sigma_1 V_0 + 
  8 \eta_0 (8 \eta_1^2 + \eta_0 \eta_2) \rho V_0^{3/2}\,,  \\
  d_2 &:= 2 \eta_0 \rho \sigma_1 \sqrt{V_0} + 16 \eta_0 \eta_1 V_0 + 16 \eta_0 \eta_1 \rho^2 V_0\,, \\
  d_3 &:= 5\eta_0 \rho \sqrt{V_0} \,.
\end{align}

The coefficients of $h_2(t)$ in \eqref{h:2:expression} have a similar expansion in $\sigma_0$ of the form:
\begin{align}\label{bar:a:b:c}
\bar a := \frac{1}{10 D^3} 3 \sigma_0 \cdot \sum_{j=0}^5 \bar c_j \sigma_0^j \,\quad
\bar b := -\frac{9}{D^2} \sigma_0 \sum_{j=0}^3 \bar d_j \sigma_0^j \,,\quad
\bar c := \frac{9}{2 D^2} \sigma_0 \sum_{j=0}^3 \bar d_j \sigma_0^j\,, 
\end{align}
where the coefficients in \eqref{bar:a:b:c} are given by:
\begin{align}
\bar c_0 &:= 32 \left(15 \eta_1^4 \rho^2 \sigma_1 V_0^2 + 13 \eta_1^5 \rho V_0^{5/2} + 
   18 \eta_0 \eta_1^3 \eta_2 \rho V_0^{5/2}\right)\,, \\
\bar c_1 &:= 16 \Big(30 \eta_1^3 \rho \sigma_1 V_0^{3/2} + 39 \eta_1^3 \rho^3 \sigma_1 V_0^{3/2} - 
   2 \eta_1^4 V_0^2 + 18 \eta_0 \eta_1^2 \eta_2 V_0^2 \\
   &\qquad\qquad+ 76 \eta_1^4 \rho^2 V_0^2 + 
   36 \eta_0 \eta_1^2 \eta_2 \rho^2 V_0^2\Big)\,, \nonumber \\
\bar c_2 &:= 8 \Big(15 \eta_1^2 \sigma_1 V_0 + 102 \eta_1^2 \rho^2 \sigma_1 V_0 + 
   66 \eta_1^3 \rho V_0^{3/2} + 36 \eta_0 \eta_1 \eta_2 \rho V_0^{3/2} \\
   &\qquad\qquad+ 
   91 \eta_1^3 \rho^3 V_0^{3/2} + 18 \eta_0 \eta_1 \eta_2 \rho^3 V_0^{3/2}\Big)\,,  \nonumber\\
\bar c_3 &:= 4 \left(87 \eta_1 \rho \sigma_1 \sqrt{V_0} + 20 \eta_1^2 V_0 + 137 \eta_1^2 \rho^2 V_0 + 
   18 \eta_0 \eta_2 \rho^2 V_0\right)\,,
 \\
\bar c_4 &:= 4 \left(12 \sigma_1 + 37 \eta_1 \rho \sqrt{V_0}\right)\,,
    \\
\bar c_5 &:= 13\,, 
\end{align}
and
\begin{align}
\bar d_0 &:= 4 \eta_1 \rho \left(3 \eta_1 \rho \sigma_1 + 2 \eta_1^2 \sqrt{V_0} + 
   2 \eta_0 \eta_2 \sqrt{V_0}\right) V_0\,, \\
\bar d_1 &:= 2 \left(7 \eta_1 \rho \sigma_1 \sqrt{V_0} + 7 \eta_1^2 \rho^2 V_0 + 
   2 \eta_0 \eta_2 \rho^2 V_0\right)\,, \\
 \bar d_2 &:=4 \sigma_1 + 7 \eta_1 \rho \sqrt{V_0}\,,  \\
 \bar d_3 &:= 1\,.
\end{align}

\end{document}